\documentclass[onecolumn,11pt,aps,nofootinbib,superscriptaddress,tightenlines]{revtex4}

\usepackage[T1]{fontenc}
\usepackage[utf8]{inputenc}

\usepackage[cmex10]{amsmath}
\usepackage{graphics}
\usepackage{dsfont}
\usepackage{amssymb}
\usepackage{graphicx}
\usepackage{mathrsfs}
\usepackage{units}
\usepackage{pgfplots}
\usepackage[colorlinks=true,linkcolor=black,citecolor=black,plainpages=false,pdfpagelabels]{hyperref}
\hypersetup{pdftitle={Spectral convergence bounds for classical and quantum Markov processes}}

\usepackage{amsthm}
\newtheorem{theorem}{Theorem}[section]
\newtheorem{lemma}[theorem]{Lemma}
\newtheorem{proposition}[theorem]{Proposition}
\newtheorem{corollary}[theorem]{Corollary}
\newtheorem{definition}{Definition}[section]

\newcommand{\abs}[1]{\ensuremath{|#1|}}

\newcommand{\norm}[2]{\ensuremath{|\!|#1|\!|_{#2}}}

\newcommand{\Norm}[2]{\ensuremath{\left|\!\left|#1\right|\!\right|_{#2}}}

\newcommand{\tr}{\textnormal{tr}}
\newcommand{\trace}[1]{\ensuremath{\tr (#1)}}

\newcommand{\braket}[2]{\langle #1 | #2 \rangle}

\renewcommand{\d}[1]{\ensuremath{\textnormal{d}#1}}

\newcommand{\iso}{\cong}

\newcommand{\id}{\ensuremath{\mathds{1}}}

\newcommand{\opid}{\ensuremath{\mathcal{I}}}

\newcommand{\linops}[1]{\ensuremath{\mathcal{L}(#1)}}

\newcommand{\cA}{\mathcal{A}}
\newcommand{\cB}{\mathcal{B}}

\newcommand{\cE}{\mathcal{E}}

\newcommand{\cH}{\mathcal{H}}

\newcommand{\cJ}{\mathcal{J}}

\newcommand{\cM}{\mathcal{M}}
\newcommand{\cN}{\mathcal{N}}

\newcommand{\cP}{\mathcal{P}}

\newcommand{\cR}{\mathcal{R}}

\newcommand{\cT}{\mathcal{T}}
\newcommand{\cU}{\mathcal{U}}
\newcommand{\cV}{\mathcal{V}}

\newcommand{\ii}{\id}

\begin{document}
\title{Spectral convergence bounds for classical and quantum Markov processes}
\author{Oleg Szehr}
\email{oleg.szehr@posteo.de}
\affiliation{Department of Mathematics, Technische Universit\"{a}t M\"{u}nchen, 85748 Garching, Germany}
\author{David~Reeb}
\email{david.reeb@tum.de}
\affiliation{Department of Mathematics, Technische Universit\"{a}t M\"{u}nchen, 85748 Garching, Germany}
\author{Michael~M.~Wolf}
\email{m.wolf@tum.de}
\affiliation{Department of Mathematics, Technische Universit\"{a}t M\"{u}nchen, 85748 Garching, Germany}

\date{\today}

\begin{abstract}  
We introduce a new framework that yields spectral bounds on norms of functions of transition maps for finite, homogeneous Markov chains. The techniques employed work for bounded semigroups, in particular for classical as well as for quantum Markov chains and they do not require additional assumptions like detailed balance, irreducibility or aperiodicity. We use the method in order to derive convergence bounds that improve significantly upon known spectral bounds. The core technical observation is that power-boundedness of transition maps of Markov chains enables a Wiener algebra functional calculus in order to upper bound any norm of any holomorphic function of the transition map. Finally, we discuss how general detailed balance conditions for quantum Markov processes lead to spectral convergence bounds.
\end{abstract}

\maketitle
\tableofcontents
\section{Introduction}
Across scientific disciplines, Markov chains are ubiquitous in algorithms as well as in models for time evolutions. In many cases one is interested in when their limit behavior is setting in. For algorithms this is often necessary in order to extract the right information and for time evolutions of physical systems this is the time scale on which relaxation or equilibration takes place. Some of the most widespread tools for bounding this time scale are based on the spectrum of the transition map. For time-homogeneous Markov chains with finite state space, the transition map is a stochastic matrix in the context of classical probability distributions and a completely positive trace-preserving map in the quantum case. Since these maps have spectral radius equal to $1$, it is somehow clear that only eigenvalues of magnitude $1$ survive the limit, that the largest subdominant eigenvalue governs the speed of convergence, and that the rest of the spectrum only matters on shorter time scales. 
Let $\cT$ and $\cT_\infty$ be the transition map and its asymptotic part, respectively. We seek convergence estimates of the form
\begin{align}
\Norm{\cT^n-\cT_\infty^n}{}\leq K\mu^{n}\label{firstbound}
\end{align}
after $n$ time steps, where $\mu$ is the magnitude of the largest eigenvalue of $\cT$ inside the open unit disc and $K$ depends on the spectrum of $\cT$, on $n$ and on the dimension of the underlying space. We demand that the dependence of $K$ on $n$ is not exponential, capturing the intuition that the convergence is determined by an exponential decay as $\mu^n$ at larger timescale, while for smaller $n$ the whole spectral data is relevant. Such bounds are of general interest for the theory of Markov chains, and they are especially important for stochastic algorithms, which are widely used in statistics and computer science. They are related to the sensitivity of the chain to perturbations \cite{Mit1,Mit2,wir}, are used to study ``cut-off\rq\rq{} phenomena \cite{cutoff} and random walks on groups \cite{walk}. More generally, the main innovation of this article lies in the development of \emph{a framework that yields spectral estimates in the context of Markov chains.}

Before describing our main results, we mention two traditional, linear algebraic, approaches to bounding convergence times of classical Markov chains as in \eqref{firstbound}. A Jordan decomposition of the difference $\cT-\cT_\infty$ yields a bound of the form Equation~\eqref{firstbound} with {$K=k\:\mu^{-d_\mu+1}n^{d_\mu-1}$}, where $d_\mu$ is the size of the largest Jordan block corresponding to any eigenvalue of magnitude $\mu$ and $k$ is constant with respect to $n$ but depends on $\cT$ as it is essentially the condition number of the similarity transformation to Jordan normal form. Unfortunately, there is no a priori bound on this condition number.
An alternative way is to use Schur's instead of Jordan's normal form. This leads indeed to an expression as in Equation~\eqref{firstbound} where $K$ can be bounded independent of $\cT$, albeit not of $n$, and we obtain roughly $K\sim \mu^{-D+1}(Dn)^D$, where $D$ is the dimension of the underlying vector space. (See Section~\ref{schurjordan} for details.) Needless to say, this ``constant'' seems to be far from optimal, especially it does not capture the (correct) asymptotic $n$-dependence of the Jordan bound.

When proving bounds of the form of Equation~\eqref{firstbound}, one typically employs additional properties of the Markov chain such as detailed balance, irreducibility, aperiodicity, uniqueness of the fixed point, Gibbs distribution of the stationary state, etc.. Clearly, these assumptions are not always fulfilled---in particular in the quantum context detailed balance seems to be a less natural assumption and, furthermore, especially in the area of dissipative quantum computing \cite{VWC2008} and dissipative state preparation \cite{DiehlNature,VWC2008,KrausPRA}, one aims at preparing rank deficient states.

For classical Markov chains convergence estimates have been widely studied \cite{seneta,mixing} and estimates based on the Jordan and Schur decompositions have been known for many years. Although the latter are generally referred to as spectral convergence bounds, they do not provide a satisfactory spectral description of the convergence of a Markov chain. While in case of the Jordan bound it is not possible to compute $K$ in terms of the eigenvalues of $\cT$, the Schur bound cannot provide the correct asymptotic behavior and does not reflect the full spectral structure of $\cT$. So far there is no a priori estimate as in Inequality~\eqref{firstbound} such that $K$ can \emph{simply be inferred from the localization of the eigenvalues} of $\cT$ and such that one obtains the correct asymptotic behavior of the chain. One goal of the present work is to close this gap and to understand what information the spectrum of the transition map of a classical or quantum Markov chain carries about the speed at which it approaches its stationary behavior, i.e., to determine $K$ in terms of the spectrum of $\cT$.

Our primary interest lies in the study of classical and quantum Markovian evolutions. However, to obtain a unified picture, in this article we will state our results more generally for bounded semigroups of linear maps. Any endomorphism $\cT$ of a vector space $\cV$ naturally generates a semigroup consisting of all $n$-fold concatenations $\cT^n$, $n\in\mathbb{N}$. In our analysis we shall assume that the vector space of endomorphisms of $\cV$ carries a norm and that the map $\cT$ is \emph{power-bounded} with respect to that norm. That means there is a constant $C$ such that, for any $n$, $\Norm{\cT^n}{}$ is bounded by $C$. This is equivalent to saying that the semigroup $(\cT^n)_{n\geq0}$ generated by $\cT$ is bounded. The framework of bounded semigroups naturally incorporates both classical and quantum Markov chains (see Section~\ref{clquMach}).

We start our discussion by analyzing the asymptotic behavior of a bounded semigroup $(\cT^n)_{n\geq0}$. We discuss spectrum related properties of $\cT$ that generate a bounded semigroup and define the asymptotic part of the evolution $\cT_\infty$ in Section~\ref{limit}. In Section~\ref{schurjordan} we extend the known convergence estimates based on the Jordan and Schur decompositions to cope with bounded semigroups. Implicitly, the analysis covers quantum Markov processes, where we state new convergence estimates.
Section~\ref{secmain} contains our main result, a convergence estimate with the form of Equation~\eqref{firstbound}, where $K$ is fully determined by $n$ and the spectrum of $\cT$. We start Section~\ref{secmain} with a mathematical primer, Section~\ref{ff}, containing an introduction to an entirely new mathematical toolbox in the context of Markov processes. We proceed by analyzing what information can be inferred from the spectrum of $\cT$ about the speed at which $(\cT^n)_{n\geq0}$ approaches its asymptotic behavior, Section~\ref{semi}. The methods, which we employ, enable us in principle to derive spectral bounds on norms of arbitrary functions of transition maps. When applied to power functions, we basically obtain the sought convergence bounds. We discuss in Subsection~\ref{discussion} how our new bound outperforms the convergence estimates based on the Jordan and Schur decompositions.

Nevertheless, it turns out that for many application, such as dissipative quantum computation and state preparation, the convergence estimates obtained still are insufficient to prove the efficiency of a possible implementation. The problem is that the convergence time grows with $D$, which in turn is exponential in the number of constituent particles (Section~\ref{discussion}). We discuss aspects related to the optimality of our new estimate as well as the convergence speed of contractive Hilbert space semigroups in Section~\ref{contractive}. We prove that stronger estimates, i.e. estimates such that roughly $\log{(D)}$ time steps bring the chain close to stationarity cannot rely on the spectrum of the transition map alone, the latter simply does not contain sufficient information.

As an approach to better convergence estimates in Section~\ref{secdetbal} we extend the detailed balance condition for classical Markov chains and define this property in the context of bounded semigroups, which then includes quantum evolutions. The core theorem of this section is an extension of a convergence estimate that is frequently used to prove cut-off behavior for classical Markov chains (Section~\ref{llsec}).

Our discussion focuses on general bounded semigroups but the corresponding statements about classical and quantum Markov processes are implicit, and we will frequently use these for illustration. In what follows one can think of $\cT$ either as a quantum channel or an ordinary stochastic matrix.

\section{Preliminaries}\label{PREL}
\subsection{Bounded Semigroups}
\label{prel:not}

Throughout this paper $\cV$ will be a real or complex vector space of finite dimension $D$. The set of linear endomorphisms of $\cV$ will be denoted by $\linops{\cV}$, which shall be endowed with a norm $\Norm{\cdot}{}$.
For a given $\cT\in\linops{\cV}$ we consider the semigroup $(\cT^n)_{n\geq0}=\{ \cT^n\:| n\in\mathbb{N}\}$ of linear maps on $\cV$ generated by $\cT$. Throughout, we assume that $(\cT^n)_{n\geq0}$ is bounded, i.e., there is a constant $C>0$ such that 
$\sup_{n\geq0}\Norm{\cT^n}{}\leq C$. 

Our main approach applies for $(\cT^n)_{n\geq0}$ with a general norm. Nevertheless, for certain results concerning convergence of classical and quantum Markov chains it will be convenient to endow $\cV$ with a scalar product $\braket{\cdot}{\cdot}$. We will consider the induced Hilbert space norm (shortly, 2-norm) $\Norm{v}{2}=\sqrt{\braket{v}{v}}$ and the operator norm (shortly, $\infty$-norm) on $\linops{\cV}$ defined by $\Norm{\cT}{\infty}=\sup_{v\neq0}\frac{\norm{\cT(v)}{2}}{\norm{v}{2}}$.
In some of our examples (e.g., classical Markov chains) it is useful to fix an orthonormal basis $\{e_i\}_{i=1,...,D}$ for $\cV$. In this case we write $T$ for the matrix representation of $\cT$ with respect to $\{e_i\}_{i}$, i.e., $T_{ij}= \braket{e_i}{ \cT(e_j)}$. We will emphasize whether or not $\cV$ has such additional structure in the corresponding sections.

\subsection{Classical and Quantum Markov chains}\label{clquMach}

We briefly review the definitions of classical and quantum Markov chains and discuss certain related concepts.

A classical, finite and time homogeneous Markov process is characterized by a semigroup generated by a classical stochastic matrix $T$. More precisely, in this scenario $\cV\iso\mathbb{R}^D$ equipped with the canonical basis $\{e_i\}_{i}$ and standard scalar product. The assertion that $T$ is stochastic is equivalent to  $T_{ij}\geq0$ and $\sum_iT_{ij}=1$. The latter is equivalent to saying that the vector $e=\sum_{i=1}^De_i\in\cV$ is fixed by the adjoint map, $T^*(e)=e$.
In the context of classical Markov chains the 1-norm plays an exceptional role. For $v\in\cV$ we write $v_i=\braket{e_i}{v}$ and define $\Norm{v}{1}=\sum_{i=1}^D{\abs{v_i}}$. The induced norm on the set of  matrices $M$ acting on $\cV$ is called the 1-to-1 norm,
$$\norm{M}{1\rightarrow1}=\sup_{v\neq0}\frac{\norm{Mv}{1}}{\norm{v}{1}}.$$
The 1-to-1 norm and the $\infty$-norm (i.e. the 2-to-2 norm) are equivalent with
\begin{align}
D^{-1/2}\Norm{M}{\infty}\leq\Norm{M}{1\rightarrow1}\leq D^{1/2}\Norm{M}{\infty}\label{clequival}.
\end{align}
It is easily seen that $\Norm{T}{1\rightarrow1}=1$ for any stochastic matrix $T$. We note that if $\Norm{\cdot}{}$ is any norm such that $\Norm{T}{}\leq C$ holds for all stochastic matrices $T$, then $\Norm{T^n}{}\leq C\ \forall n\in\mathbb{N}$; that is the Markov chain constitutes a bounded semigroup with constant $C$. Since we are working in finite dimensions, such a semigroup is bounded with respect to any norm.

A time homogeneous quantum Markov chain is also characterized by a semigroup. In the context of quantum evolutions, however, the space $\cV$ has different and additional structure. In this article we think of $\cV$ as the real vector space consisting of Hermitian matrices acting on a complex Hilbert space of dimension $d$, i.e., $D=d^2$. A matrix $\rho\in\cV$ that is positive semidefinite ($\rho\geq0$) and has unit trace ($\tr{[\rho]}=1$) is referred to as a quantum state. An element $\cT\in\linops{\cV}$ is called positive iff $X\geq0$ implies $\cT(X)\geq0$ for any $X\in\cV$, and trace-preserving iff $\tr{[\cT(X)]}=\tr{[X]}\ \forall X\in\cV$. $\cT$ is trace-preserving iff the adjoint $\cT^*$ of $\cT$ with respect to the Hilbert-Schmidt inner product $\braket{X}{Y}=\trace{XY}$ on $\cV$ preserves the identity matrix, $\cT^*(\ii)=\ii$. If $\cT\otimes\opid$, with $\opid$ being the operator identity, acts as a positive map on $\cV\otimes\cV$, then $\cT$ is called completely positive \cite{NielsenChuang, Paulsen}.
We denote by $\mathfrak{T}$ the subset of $\linops{\cV}$ containing trace preserving and positive maps (TPPMs) and by $\mathfrak{T}_+\subset\mathfrak{T}$ the set of completely positive maps in $\mathfrak{T}$ (TPCPMs). The latter describe the dynamics of a quantum system, whenever the evolution of the system is independent of its history, and they are called quantum channels in the realm of quantum information theory.

For $X\in\cV$ we denote by $\Norm{X}{1}$ the Schatten 1-norm of $X$. The induced distance $\Norm{\rho-\sigma}{1}$ of two quantum states $\rho$ and $\sigma$ corresponds to the maximum probability to detect a difference between $\rho$ and $\sigma$ in an experiment, i.e.
\begin{align*}
\Norm{\rho-\sigma}{1}=\sup_{\Norm{O}{\infty}\leq1}\abs{\trace{O(\rho-\sigma)}},
\end{align*}
where $\Norm{O}{\infty}$ stands for the largest singular value of $O\in\cV$. For linear maps $\cM\in\linops{\cV}$ we define\footnote{Note that in this article we define the 1-to-1 norm with the supremum taken over Hermitian matrices. Alternatively, the supremum could be taken over all matrices. The resulting norms are different, but the latter can be upper bounded in terms of 2 times the former.} the induced 1-to-1-norm via
$$\norm{\cM}{1\rightarrow1}=\sup_{X\neq0}\frac{\norm{\cM(X)}{1}}{\norm{X}{1}}.$$
The diamond norm is the ``stabilized version\rq\rq{} of the 1-to-1 norm,
$$\norm{\cM}{\diamond}=\Norm{\cM\otimes\opid}{1\rightarrow1},$$
where $\opid$ denotes the operator identity in $\linops{\cV}$. It is the dual of the norm of complete boundedness (CB-norm), i.e., we have $\Norm{\cM}{\diamond}=\Norm{\cM^*}{CB}$. The diamond norm and the 1-to-1 norm are equivalent with \cite{smith}
\begin{align}
\Norm{\cM}{1\rightarrow1}\leq\Norm{\cM}{\diamond}\leq D^{1/2}\Norm{\cM}{1\rightarrow1}.\label{diamond1to1}
\end{align}
For any quantum channel $\cT$ we have $\norm{\cT}{1\rightarrow1}=\norm{\cT}{\diamond}=1$ \cite{Paulsen}. The distance $\Norm{\cE-\cT}{\diamond}$ of two channels $\cE,\cT$ measures how well these channels can be distinguished by any quantum experiment. In the quantum context the 1-to-1 norm and the $\infty$-norm (i.e., the 2-to-2 norm) are equivalent with
\begin{align}
D^{-1/4}\Norm{\cT}{\infty}\leq\Norm{\cT}{1\rightarrow1}\leq D^{1/4}\Norm{\cT}{\infty}\label{equival}.
\end{align}
Note that due to the different structure of the underlying space $\cV$ in case of quantum Markov chains, the above differs from the Inequalities~\eqref{clequival}.

If we are given a norm such that $\norm{\cT}{}\leq C\ \forall \cT\in\mathfrak{T}$ (or $\forall\cT\in\mathfrak{T}_+$) the quantum Markov chain generated by $\cT$ constitutes a semigroup bounded by $C$. Again, due to $D<\infty$, this implies that the semigroup is bounded with respect to any norm.
%
%
%
\subsection{Spectral properties}
To each linear map $\cM\in\linops{\cV}$ we can assign a spectrum $\sigma(\cM)$ via the usual eigenvalue equation: we have $\lambda\in\sigma(\cM)$ if and only if there is $X\neq0$ with $\cM(X)=\lambda X$. We write $m_\cM$ for the minimal polynomial associated with $\cM$ (i.e.,~the minimal degree, monic polynomial that annihilates $\cM$, $m_\cM(\cM)=0$) and $\abs{m_\cM}$ for the number of linear factors in $m_\cM$. Another important object is the Blaschke product associated with $m_\cM$,
\begin{align}
B(z)=\prod_{m_\cM}\frac{z-\lambda_i}{1-\bar{\lambda}_i z},\label{bp}
\end{align}
where the product is taken over all $i$ such that the linear factor $z-\lambda_i$ occurs in $m_\cM$ \emph{respecting} multiplicities. Thus, the numerator of $B$ as defined here is exactly the corresponding minimal polynomial, $m_\cM$. 

For convenience, we shall always assume that the eigenvalues in $\sigma(\cM)$ are arranged such that their magnitudes are non-decreasing. (This ordering is not unique when several eigenvalues have the same magnitude. This ambiguity will, however, be irrelevant in the following. Whenever the situation occurs that we pick an eigenvalue of a certain magnitude $\abs{\lambda}$ we mean that we can take \emph{any} eigenvalue that has this property.)

For any $\cM\in\linops{\cV}$ the largest magnitude of all eigenvalues is the spectral radius, which we denote as $\mu$. It follows from Gelfand\rq{}s formula $\mu=\lim_{k\rightarrow\infty} \Norm{\cM^k}{}^{1/k}$ \cite{specevans} that the spectral radius of any element of a bounded semigroup is at most $1$. For stochastic matrices and TPPM it is clear that $1$ is an eigenvalue of $\cT$.
\section{{Limiting behavior and classical convergence estimates}}
\subsection{Limiting behavior of $(\cT^n)_{n\geq0}$}\label{limit}
%
%
%
In this section we begin our discussion of spectral convergence bounds for semigroups. Based on a spectral decomposition of $\cT\in(\cT^n)_{n\geq0}$ we introduce a map $\cT_\infty$ and show that this map reflects the behavior of $(\cT^n)_{n\geq0}$ for large $n$. 
In the following we extend known spectral convergence bounds for classical Markov chains to the more general semigroup setup. We consider the classical derivations based on the Jordan and Schur decomposition (Section~\ref{schurjordan}). For this reason in this section we assume that $\cV$ carries a scalar product.


Our main result Theorem~\ref{result} will later outperform the bounds proven in this section in terms of convergence speed even in the context of classical Markov chains. Moreover, the techniques introduced there will allow us to consider general norms, which are not induced by a scalar product.

To formalize our intuition that the spectrum of $\cT$ determines the convergence properties of $(\cT^n)_{n\geq0}$ let us consider a Jordan decomposition of $\cT$,
\begin{eqnarray}
\cT&=&\sum_i\left({\lambda_i}\cP_i+\cN_i\right), \quad\mbox{with}\\
&& \cN_i\cP_i=\cP_i\cN_i=\cN_i,\ \ \cP_i\cP_j=\delta_{i,j}\cP_i\ \forall i,j.\label{eq:Jordanrelations}
\end{eqnarray}
Here, the summation is taken over all distinct eigenvalues of $\cT$, the $\cP_i$ are projectors whose rank equals the algebraic multiplicity of $\lambda_i$ and  the $\cN_i$ are the corresponding nilpotent blocks.
All contributions to $\cT^n$ that stem from eigenvalues of $\cT$ with magnitude smaller than $1$ will vanish with increasing $n$. Hence, we expect the image of $\cT^n$ to converge to a subspace of $\cV$ spanned by all eigenvectors of $\cT$ whose eigenvalues are of magnitude one. We therefore define the linear map $\cT_\infty$ whose range is this subspace by
\begin{align}
\cT_\infty:=\sum_{\abs{\lambda_i}=1}{\lambda_i}\cP_i,\label{Tinfty}
\end{align}
where the $\cP_i$ are spectral projectors corresponding to the eigenvalues of $\cT$ of magnitude $1$.
In cases where the spectral radius of $\cT$ is strictly smaller than $1$, $\cT_\infty$ is simply zero.
If $\cT$ has only one eigenvalue of magnitude one and this eigenvalue is equals $1$, then the sequence $\cT^n$ converges to $\cT_\infty$, which is the unique rank one projection onto the stationary eigenspace of $\cT$.
In the following lemma we shall prove that $\cT_\infty$ mirrors the limit behavior of $(\cT^n)_{n\geq0}$ also in the more general case. More precisely,  as $n$ is increasing $\Norm{\cT^n-\cT_\infty^n}{}$ approaches $0$ and, for every $k\in\mathbb{N}$, $\cT_\infty^k$ indeed is an accumulation point of $(\cT^n)_{n\geq0}$. %
The latter assertion is relevant especially in the case of classical and quantum Markov chains: the set of stochastic matrices (or quantum channels) constitutes a closed set in the corresponding space, which implies that $\cT_\infty$ is again a bona fide stochastic matrix (or quantum channel).
\begin{lemma}[Limiting behavior of $\cT^n$]\label{Tunendlich} Let $(\cT^n)_{n\geq0}$ be a semigroup within $\linops{\cV}$ such that $\Norm{\cT^n}{}\leq C\ \forall n\in\mathbb{N}$ and let $\cT_\infty$ be as in Equation~\eqref{Tinfty}. Then we have that\\
i) all eigenvalues of $\cT$ with magnitude $1$ have trivial Jordan blocks (i.e., $\abs{\lambda_i}=1 \Rightarrow\cN_i=0$),\\
ii) $(\cT^n-\cT_\infty^n)=(\cT-\cT_\infty)^n\ \forall\ n\in\mathbb{N}\backslash\{0\}$,\\
iii) $\lim_{n\rightarrow\infty}\Norm{\cT^n-\cT_\infty^n}{}=0$,\\
iv)  for any $k\in\mathbb{N}$, $\cT_\infty^k$ is contained in the closure of $(\cT^n)_{n\geq0}$ in $\linops{\cV}$,\\
v) $\Norm{\cT_\infty^k}{}\leq C\ \forall k\in\mathbb{N}$.
\end{lemma}
\begin{proof}
i) We proceed by contradiction and consider $\Norm{\cT^n}{\infty}$. Since $(\cT^n)_{n\geq0}$ is bounded and in finite dimensions all norms are equivalent there is $0<K_1<\infty$ with $\Norm{\cT^n}{\infty}\leq K_1$. On the other hand there is $K_2>0$ with $\Norm{\cT^n}{\infty}\geq K_2\Norm{\sum_j (\lambda_j \cP_j+\cN_j)^n}{\infty}$. If $\lambda_i$ has a non-trivial Jordan block the latter can be lower bounded by $\Norm{\sum_j (\lambda_j \cP_j+\cN_j)^n}{\infty}\geq\abs{\lambda_i}^{n-1} n$. It follows that if $\abs{\lambda_i}=1$ and $\lambda_i$ has a non-trivial Jordan block then $\Norm{\cT^n}{\infty}$ grows unboundedly with $n$.\footnote{See also the derivation of the lower bound in Theorem~\ref{jordan}.}\\
ii) follows from the relations in (\ref{eq:Jordanrelations}) since $\cT\cT_\infty=\cT_\infty\cT=\cT_\infty^2$. For $n>2$ the statement follows by induction.\\
iii) 
By the previous assertion $\Norm{T^n-T_\infty^n}{}=\Norm{(T-T_\infty)^n}{}$ holds. The spectral radius $\mu$ of the map $\cT-\cT_\infty$ is strictly smaller than $1$. We have from Gelfand\rq{}s formula that $\lim_{n\rightarrow\infty}\Norm{T^n-T_\infty^n}{}^{1/n}=\mu<1$ and hence for all $n$ sufficiently large $\Norm{T^n-T_\infty^n}{}\leq\left(\frac{1+\mu}{2}\right)^n$. With increasing $n$ the right hand side goes to $0$ and the claim follows.\\
iv) We prove that for fixed $k$ there is a subsequence $(T^{n_l})_l$ that converges to $\cT_\infty^k$, that is $\lim_{l\rightarrow\infty}\Norm{\cT^{n_l}-\cT_\infty^k}{}=0$. To achieve this we subdivide $\cV$ into the invariant subspace of $\cT$ corresponding to all eigenvalues of magnitude $1$ and its complement. On the latter subspace we can directly invoke iii). On the former subspace, it is sufficient to find, for any $\epsilon>0$, a subsequence $(T^{n_l})_l$ with the property that $\abs{\lambda_i^{n_l}-\lambda_i^k}\leq\epsilon$ simultaneously for all $i$ with $\abs{\lambda_i}=1$. The existence of such a subsequence follows from Dirichlet\rq{}s Theorem on simultaneous Diophantine approximation, \cite{diophantine} Theorem 1B.\\
v) By iv) for any $k\in\mathbb{N}$ and any $\epsilon>0$ there is $n$ such that $\Norm{\cT^n-\cT_\infty^k}{}\leq\epsilon$. This implies that $\Norm{\cT_\infty^k}{}\leq C+\epsilon\ \forall\epsilon>0$ and hence the claimed inequality.
\end{proof}

\subsection{Jordan and Schur convergence estimates}\label{schurjordan}

Our next aim is to understand qualitatively by how much for certain $n$ the evolution $\cT^n$ differs from its limit behavior, i.e., how small the quantity $\Norm{\cT^n-\cT_\infty^n}{\infty}=\Norm{T^n-T_\infty^n}{\infty}$ is for any bounded semigroup. We shortly review two standard methods to obtain such estimates. Both methods rely on the fact that $\cT^n-\cT_\infty^n=(\cT-\cT_\infty)^n$ and perform a transformation of $T-T_\infty$ to upper triangular form. While the first approach is to choose the Jordan normal form for $T-T_\infty$, the second one is based on the Schur decomposition. Both  decompositions involve a similarity transformation $A$ that brings $T-T_\infty$ to upper triangular form, i.e., $T-T_\infty=A(\Lambda+N)A^{-1}$ with diagonal $\Lambda$ and nilpotent $N$. While in case of Jordan decomposition $\Lambda+N$ has Jordan block structure, for the Schur decomposition $A$ is unitary.
\begin{theorem}\label{jordan}
Let $(\cT^n)_{n\geq0}$ be a bounded semigroup in $\linops{\cV}$, let $\cT_\infty$ be the map introduced in \eqref{Tinfty} and let $\mu$ be the spectral radius of $\cT-\cT_\infty$. Then there are constants $C_1, C_2>0$ such that, for all $n\geq1$,
\begin{align*}
C_1\mu^{n-d_\mu+1} n^{d_\mu-1}\leq\Norm{\cT^n-\cT_\infty^n}{\infty}\leq C_2\mu^{n-d_\mu+1} n^{d_\mu-1},
\end{align*}
where $d_\mu$ is the size of the largest Jordan block corresponding to any eigenvalue of $\cT-\cT_\infty$ of magnitude $\mu$. 
\end{theorem}
\begin{proof}
We first state an upper bound on $\Norm{(\Lambda+N)^n}{\infty}$ with diagonal $\Lambda$ and nilpotent upper-triangular $N$. We note that any monomial in $N$ and $\Lambda$ vanishes if the total degree of  $N$ is larger than or equal to $j D-1$. Using this together with the triangle inequality in the binomial expansion and exploiting the sub-multiplicativity of the $\infty$-norm we find  
\begin{align}\label{binary}
\Norm{(\Lambda+N)^n}{\infty}\leq\sum_{k=0}^{\min{\{n,D-1}\}}
{n\choose k}
\Norm{N}{\infty}^k\Norm{\Lambda}{\infty}^{n-k}.
\end{align}

Let now $J(\lambda_i)$ be a Jordan block with diagonal part $\lambda_i\id$ and nilpotent part $N_i$. We consider the Jordan decomposition $T-T_\infty=A\left(\bigoplus_{i,\nu}J_\nu(\lambda_i)\right)A^{-1}$, where the summation goes over $i$, which labels the different eigenvalues of $T-T_\infty$, and over $\nu$, which enumerates the Jordan blocks corresponding to an eigenvalue $\lambda_i$. We introduce the constant $\kappa=\inf{\left(\Norm{A}{\infty}\Norm{A^{-1}}{\infty}\right)}$ where the infimum is taken over all $A$ that bring $T-T_\infty$ to Jordan form.
It follows readily that
\begin{align}\label{binary1}
\kappa^{-1}\Norm{J^n}{\infty}\leq\Norm{T^n-T^n_\infty}{\infty}\leq\kappa\Norm{J^n}{\infty}
\end{align}
with $J=\bigoplus_{i,\nu}J_\nu(\lambda_i)$. For any $J_\nu(\lambda_i)$ there is an $n_0$ such that for all $n\geq n_0$ one has $\Norm{J_\nu(\lambda_i)^n}{\infty}\leq\Norm{J_{max}(\lambda_{max})^n}{\infty}$,
where $J_{max}(\lambda_{max})$ denotes the largest Jordan block corresponding to an eigenvalue $\lambda_{max}$ of modulus $\mu$. Therefore, to find an upper bound on the right hand side of  \eqref{binary1} we can subdivide $J_{max}(\lambda_{max})$ in a nilpotent and a diagonal part and use Inequality~\eqref{binary}. We note that for $k\leq d_\mu-1$ we can bound ${n\choose k}\leq n^{d_\mu-1}$ and taking everything together we obtain for large enough $n$
\begin{align*}
\Norm{T^n-T^n_\infty}{\infty}\leq\kappa\sum_{k=0}^{d_\mu-1}n^{d_\mu-1}
\mu^{n-d_\mu+1},
\end{align*}
which proves the upper bound in Theorem~\ref{jordan} since it can be extended to an upper bound valid for any $n\in\mathbb{N}$ by a rescaling of $C_1$. 
The lower bound is a consequence of the following inequalities for $n\geq d_\mu-1$
\begin{align*}
\mu^{n-d_\mu+1}
{n\choose d_\mu-1}\leq\Norm{J_{max}(\lambda_{max})^n}{\infty}\leq\Norm{J^n}{\infty}.
\end{align*}
\end{proof}
One problem with the above proof is that $n_0$ and thus $C_2$ can get large if there is a sub-dominant eigenvalue close to the spectral radius. Another issue is that one cannot a priori bound $\kappa$ for general $T$. Consequently, only little is known about $C_1$ and $C_2$. Most awkward, $C_1$ and $C_2$ depend on the given channel $\cT$, i.e. are not universal for all channels of a given dimension. For this reason Theorem~\ref{jordan} is a qualitative statement about the asymptotic behavior of the semigroup. In contrast, the Schur decomposition allows us to state an upper bound on the rate of convergence that only depends on $n$, $D$ and $\mu$. This goes at the price of a rather pessimistic estimate.

\begin{theorem}\label{schur}
Let $(\cT^n)_{n\geq0}$ be a bounded semigroup in $\linops{\cV}$ such that $\norm{\cT^n}{\infty}\leq C\ \forall n\in\mathbb{N}$ and let $\mu$ be the spectral radius of $\cT-\cT_\infty$. For any $n\in\mathbb{N}$ it holds that
\begin{align*}
\Norm{\cT^n-\cT_\infty^n}{\infty}\leq2\mu^{n-D+1}n^{D-1}(\mu+2C)^{D-1}.
\end{align*}
\end{theorem}
\begin{proof}
As already mentioned, this will be proven based on the Schur decomposition $T-T_\infty=U(\Lambda+N)U^\dagger$, where $U$ is unitary. As before we can rely on the binomial expansion Inequality~\eqref{binary}. We note that $\Norm{U}{\infty}=1$ and that for $n>1$
\begin{align*}
\sum_{k=0}^{D-1}{n \choose k}\leq\sum_{k=0}^{D-1}n^k\leq2n^{D-1}.
\end{align*}
Thus, using the sub multiplicativity of the $\infty$-norm it follows from \eqref{binary} that
\begin{align*}
\Norm{T^n-T^n_\infty}{\infty}
\leq2n^{D-1}\mu^{n-D+1}\max{(1,\Norm{N}{\infty}^{D-1})}.
\end{align*}
In addition we have that $N=T-T_\infty-\Lambda$ and therefore $\Norm{N}{\infty}\leq2C+\mu$.
\end{proof}
To obtain a convergence estimate for Markov chains in 1-to-1 norm we can rely on the Inequalities~\eqref{clequival}. The corresponding statement of Theorem~\ref{jordan} is immediate. Analogously, Theorem~\ref{jordan} can be used to estimate the speed of convergence of TPPMs in 1-to-1 and diamond norm via the Inequalities~\eqref{equival},~\eqref{diamond1to1}.

Due to the lower bound in \eqref{clequival} the singular values of stochastic matrices are bounded by $D^{1/2}$ from which we infer that $C\leq D^{1/2}$ in this case. For positive, trace preserving maps the singular values are bounded by $D^{1/4}$ (\cite{RuWo}, or by the norm equivalence \eqref{equival} and the fact that $\Norm{\cT}{1\rightarrow1}=1$). Thus, Theorem~\ref{schur} includes a convergence bound for both classical stochastic matrices and TPPMs. For a more detailed discussion of the resulting estimates in the quantum context see Subsection~\ref{discussion}.
\section{Main result: Spectrum and Convergence}\label{secmain}
The main contribution of this article is to introduce a new formalism that yields spectral bounds on norms of functions of transition maps of Markov processes and to apply this formalism to prove new estimates for the convergence of a such  processes to stationarity.
The core technical innovation will be to employ a Wiener algebra functional calculus in the context of bounded semigroups. To prove our estimates we will rely on the theory of function algebras, functional calculi and model spaces. To our knowledge these concepts have not found their way into the theory of classical or quantum Markov processes so far. For this reason at first we briefly introduce the mathematical framework in Subsection~\ref{ff}. A detailed introduction to the mathematics involved goes beyond the scope of this article and we refer to \cite{TS,OF,AC} for this. In Subsection~\ref{semi} we employ the mathematical machinery to the context of bounded semigroups and derive the main theorem. In the subsequent subsection we discuss our main result and compare it to the convergence estimates from Jordan and Schur decompositions.

\subsection{Function spaces and functional calculi}\label{ff}
In this subsection we discuss the problem of bounding the norm of a function of an operator in terms of the spectrum of this operator. To start with, we introduce the classes of functions and operators that we study and recall the notion of a bounded functional calculus. In the mathematical literature the problem of constructing good functional calculi for given classes of operators is studied extensively \cite{Nagy,TS,OF}. The boundedness of the functional calculus implies that the norm of a function of the operator is bounded in terms of the norm of the function. A core innovation taken from~\cite{AC} is then to relate the problem of finding a good spectral bound to a Nevanlinna-Pick interpolation problem in the corresponding class of functions.

We begin by defining the function spaces, which will be relevant in our discussion. The space of analytic functions on the open unit disc $\mathbb{D}=\{z\in\mathbb{C}|\abs{z}<1\}$ is denoted by $Hol(\mathbb{D})$. We will be concerned with certain subspaces of $Hol(\mathbb{D})$, an important class of which constitute the Hardy spaces. For $p>0$ those are defined as
\begin{align*}
H^p:=\big\{f\in Hol(\mathbb{D})| \Norm{f}{H^p}^p:=\sup_{0\leq r<1}\frac{1}{2\pi}\int_0^{2\pi}\abs{f(re^{i\phi})}^p\d{\phi}<\infty\big\},
\end{align*}
and 
\begin{align*}
&H^\infty:=\big\{f\in Hol(\mathbb{D})| \Norm{f}{H^\infty}:=\sup_{z\in\mathbb{D}}\abs{f(z)}<\infty\big\}.
\end{align*}
It is immediate from the definition that the spaces $H^p$ are vector spaces, that the mapping $f\mapsto\Norm{f}{H^p}$ is a norm for $p\geq1$ and that $H^p\subset H^q$ for $p\geq q$. In the special case $p=2$ the Hardy norm can be written using the Taylor coefficients of the analytic function $f$. More precisely, we write $f(z)=\sum_{k\geq0}\hat{f}(k)z^k$ and use Parseval\rq{}s identity to conclude that
\begin{align*}
\sup_{0\leq r<1}\frac{1}{2\pi}\int_0^{2\pi}\abs{f(re^{i\phi})}^2\d{\phi}=\sum_{k\geq0}\abs{\hat{f}(k)}^2.
\end{align*}
Thus, $f\in Hol(\mathbb{D})$ is in $H^2$ if and only if $\sum_{k\geq0}\abs{\hat{f}(k)}^2<\infty$ (see \cite{OF}, p. 32). The Wiener algebra is defined as the subset of $Hol(\mathbb{D})$ of absolutely convergent Taylor series,
\begin{align*}
W:=\{f=\sum_{k\geq0}\hat{f}(k)z^k|\sum_{k\geq0}\abs{\hat{f}(k)}<\infty\}.
\end{align*}

For a given class of operators (for instance Hilbert space contractions or power bounded operators) the associated function algebra is a space of analytic functions that  mirrors the ``boundedness properties\rq\rq{} of those operators. A functional calculus is a map that associates operators from the given class and elements of the function algebra and relates the norms of an operator and its representative in the function algebra.
More precisely we have the following definitions~\cite{AC}:
\begin{definition}[Function algebra]\label{functionalalgebra} A unital Banach algebra $A$ with elements in $Hol(\mathbb{D})$ will be called a function algebra, if 
\begin{enumerate}
\item $A$ contains all polynomials and $\lim_{n\rightarrow\infty}\Norm{z^n}{A}^{1/n}=1$.
\item $(a\in A,\:\lambda\in\mathbb{D},\: a(\lambda)=0)\Rightarrow\frac{a}{z-\lambda}\in A$.
\end{enumerate}
\end{definition}
\begin{definition}[Functional calculus]
Let $X: \cB\rightarrow \cB$ be an operator on a Banach space $\cB$. A bounded algebra homomorphism from a function algebra $A$ into the set of linear operators on $\cB$,
\begin{align*}
\cJ_X:\ A\rightarrow L(\cB),
\end{align*}
will be called a functional calculus for $X$, if it satisfies $\cJ_X(z)=X$ and $\cJ_X(1)=\id$.
\end{definition}
(In our case it is sufficient to assume that $\cB$ has finite dimension.)
Intuitively $\cJ_X$ captures the notion of ``plugging an operator into a function\rq\rq{}, that is for $a\in A$ we have $a(X)=\cJ_X(a)$ and by the boundedness property there is a constant $C_X$ such that
\begin{align*}
\Norm{a(X)}{}\leq C_X\Norm{a}{A}.
\end{align*}
Given a family $\Gamma$ of operators we say that this family obeys a functional calculus with constant $C$ if each $X\in\Gamma$ admits a functional calculus with $C_X\leq C$. Thus, one approach to the problem of bounding the norm $\Norm{a(X)}{}$ for $X\in\Gamma$ is by constructing a functional calculus for the family $\Gamma$ and then bounding the norm of $a$ in the function algebra.
For us, two instances of functional calculi will be important. In the first example we consider power-bounded Banach spaces operators, while the second one treats Hilbert space contractions.\\
i) Consider a family $\Gamma=\{X\in L(\cB)| \Norm{X^n}{}\leq C\ \forall n\in\mathbb{N}\}$ of Banach space operators that are power bounded by some constant $C$. This family admits a Wiener algebra functional calculus since for any $f\in W$ and $X\in\Gamma$
\begin{align}
\Norm{f(X)}{}=\Norm{\sum_{k\geq0}\hat{f}(k)X^k}{}\leq \sum_{k\geq0}\abs{\hat{f}(k)}\Norm{X^k}{}\leq C\sum_{k\geq0}\abs{\hat{f}(k)}=C\Norm{f}{W}\label{calc}
\end{align}
holds.\\
ii) In Section~\ref{contractive} we discuss the semigroup of Hilbert space contractions $\Gamma=\{X\in L(\cH)| \Norm{X}{\infty}\leq1\}$. This family allows for an $H^\infty$ functional calculus (with constant $C=1$), since by von Neumann\rq{}s inequality \cite{Paulsen} we have for any $f\in H^\infty$ that has a continuous extension to the boundary of $\mathbb{D}$ and $X\in\Gamma$
\begin{align}
\Norm{f(X)}{\infty}\leq\Norm{f}{H^\infty}\label{calc1}.
\end{align}
At first glance, the outlined procedure seems to be of little use since the right hand sides of \eqref{calc},~\eqref{calc1} do not depend on $X$ anymore. To obtain a better bound one can rely on the following insight. Recall that the minimal polynomial $m_X$ annihilates the corresponding operator, i.e., $m_X(X)=0$. Instead of considering the function $a$ directly, we add multiples of $m=m_X$ (or any other annihilating polynomial) to this function and consider $c=a+mb,\:b\in A$ instead of $a$. It is immediate that $\Norm{a(X)}{}=\Norm{c(X)}{}$. The following simple but crucial lemma summarizes this point:

{\begin{lemma}[\cite{AC}, Lemma 3.1]\label{wiener} Let $m\neq0$ be a polynomial and let $\Gamma$ be a set of operators that obey an $A$ functional calculus with constant $C$ and that satisfy $m(X)=0\ \forall X\in\Gamma$. Then
%
%
\begin{align*}
\Norm{a(X)}{}\leq C \Norm{a}{A/mA},\ \forall X\in\Gamma,
\end{align*}
where $\Norm{a}{A/mA}=\inf{\{\norm{c}{A}|\ c=f+mb,\:b\in A\}}$.
\end{lemma}}
\begin{proof} For any $b\in A$ we have that $\Norm{a(X)}{}=\Norm{(a+mb)(X)}{}\leq C \Norm{a+mb}{A}$.
\end{proof}

\subsection{Spectral bounds for the convergence of Markovian processes to stationarity}\label{semi}

Crucial for the main result Theorem~\ref{result} is that classical stochastic matrices and quantum channels both obey a power-boundedness condition. Given any norm $\Norm{\cdot}{}$ such that every $\cT\in\mathfrak{T}$ satisfies $\Norm{\cT}{}\leq C$, then for all $n\geq0$, $\Norm{\cT^n}{}\leq C$, i.e., $\cT$ generates a bounded semigroup $(\cT^n)_{n\geq0}$. In view of Lemma~\ref{wiener} this entails that $(\cT^n)_{n\geq0}$ obeys a Wiener algebra functional calculus with $\Norm{f(\cT)}{}\leq C\Norm{f}{W/mW}$. Although this observation is simple, we state it in a separate theorem to emphasize its importance.\\

\begin{theorem}\label{wienerthm} Let $(\cT^n)_{n\geq0}$ be a semigroup bounded with
constant $C$ and let $m$ be the minimal polynomial of $\cT$, $m(\cT)=0$. Then 
\begin{align*}
\Norm{f(\cT)}{}\leq C\Norm{f}{W/mW}
\end{align*}
holds for any function $f\in W$.
\end{theorem}
Theorem~\ref{wienerthm} can be used to bound various functions of transition maps of Markovian evolutions. For instance one might be interested in bounding the norm of the inverse of a transition map (if it exists). In \cite{AC} an estimate of $X^{-1}$ is derived for an algebraic Banach space operator $X$ by using Lemma~\ref{wiener} and  bounding $\Norm{z^{-1}}{W/mW}$. This estimate immediately carries over to Markov transition maps. In this article we seek bounds for the rate of convergence of a semigroup; we will use Theorem~\ref{wienerthm} to relate this problem to the one of bounding $\Norm{z^n}{W/mW}$. The latter task has not yet been studied in the mathematical literature although it is deeply connected to the famous Kreiss matrix theorem. If spectral data is present then the resolvent estimate in the Kreiss matrix theorem can be extended to the interior of the unit disk and bounding $\Norm{z^n}{W/mW}$ corresponds to the task of establishing power-boundedness with given spectrum.

Based on Theorem~\ref{wienerthm} we obtain the following:

\begin{theorem}\label{result} Let $(\cT^n)_{n\geq0}$ be a semigroup bounded by $C$, and let $\cT_\infty$ be its asymptotic evolution introduced in \eqref{Tinfty}. We write $m=m_{\cT-\cT_\infty}$ for the minimal polynomial and $\mu$ for the spectral radius of $\cT-\cT_\infty$ and $B$ for the Blaschke product \eqref{bp} associated with $m$. Then, for $n>\frac{\mu}{1-\mu}$ we have
\begin{align*}
\Norm{\cT^n-\cT_\infty^n}{}\leq\mu^{n+1}\frac{4  C e^{2}\sqrt{\abs{m}}(\abs{m}+1)}{n\left(1-(1+\frac{1}{n})\mu\right)^{3/2}}\:\sup_{\abs{z}=\mu(1+1/n)}\left|\frac{1}{B(z)}\right|.
\end{align*}
\end{theorem}
Before we proceed to the proof of Theorem~\ref{result} let us discuss some immediate consequences. First, note that the condition $n>\mu/(1-\mu)$ does not significantly restrict the range of $n$, where the theorem applies. For $n\leq\mu/(1-\mu)$ it holds that the exponentially decaying factor $\mu^n\gtrapprox e^{-\mu}$ is still of order $1$. In this range bounds of the form~\eqref{firstbound} only yield a trivial statement.

As compared to Theorem~\ref{jordan} and Theorem~\ref{schur} the bound in Theorem~\ref{result}  depends  more explicitly on the spectral properties of $\cT-\cT_\infty$. The Jordan block structure of $\cT-\cT_\infty$ is reflected by the fact that the formula contains a certain factor for each factor of $m_{\cT-\cT_\infty}$. In contrast to Theorems~\ref{jordan},~\ref{schur}, Theorem~\ref{result} clarifies in which way the Jordan structure of $\cT-\cT_\infty$ influences the speed of convergence of a Markov process.

The upper bound in Theorem~\ref{result} can be made more explicit by taking the supremum over all factors in the Blaschke product individually. It is not difficult to see (see Appendix~\ref{Blaschke}) that for $\abs{\lambda}<\mu(1+1/n)\leq1$ one has
\begin{align}
\sup_{\abs{z}=\mu(1+1/n)}\left|\frac{1-\bar{\lambda}z}{z-\lambda}\right|=\frac{1-(1+1/n)\mu\abs{\lambda}}{\mu-\abs{\lambda}+\mu/n}. \label{supbound}
\end{align}
This leads to the following corollary:
\begin{corollary}\label{corresult}
Let $\sigma(\cT-\cT_\infty)=\{\lambda_1,...,\lambda_D\}$ be the spectrum of $\cT-\cT_\infty$ so that the magnitudes are ordered non-decreasingly and let $\mu=\abs{\lambda_D}$ be the spectral radius of $\cT-\cT_\infty$. Under the assumptions of Theorem~\ref{result} it holds that
\begin{align*}
\Norm{\cT^n-\cT_\infty^n}{}
\leq\mu^{n}\frac{4  C e^{2}\sqrt{\abs{m}}(\abs{m}+1)}{\left(1-(1+\frac{1}{n})\mu\right)^{3/2}}\:\prod_{m/(z-\lambda_{D})}\frac{1-(1+\frac{1}{n})\mu\abs{\lambda_i}}{\mu-\abs{\lambda_i}+\frac{\mu}{n}},
\end{align*}
where the product is taken over all $i$ such that the corresponding linear factor $(z-\lambda_i)$ occurs in a prime factorization of $m/(z-\lambda_{D})$, \emph{respecting} multiplicities and $\lambda_{D}$ stands for \emph{any} eigenvalue of magnitude $\mu$.
\end{corollary}
Every eigenvalue of magnitude $\mu$ contributes one factor proportional to $n/\mu$ in Equation~\eqref{supbound}. Whereas Theorem~\ref{result} contains an inverse Blaschke factor for each linear factor in the minimal polynomial $m$, in Corollary~\ref{corresult} we have canceled one of the factors corresponding to the spectral radius $\mu$ by the $\mu/n$ prefactor in Theorem~\ref{result}.

The techniques upon which the derivation of Theorem~\ref{result} builds also yield more general \emph{geometric} convergence estimates, where the exponentially decaying factor $\mu^n$ is replaced by $\beta^n$ for some $\beta>\mu$. In this case the prefactor can be chosen independent of $n$.

\begin{corollary}\label{geometric}
Let $(\cT^n)_{n\geq0}$ be a semigroup bounded by $C$, and let $\cT_\infty$ be its asymptotic evolution introduced in \eqref{Tinfty}. We write $m=m_{\cT-\cT_\infty}$ for the minimal polynomial and $\mu$ for the spectral radius of $\cT-\cT_\infty$ and $B$ for the Blaschke product \eqref{bp} associated with $m$. Then, for any $\beta\in(\mu,1)$ we have
\begin{align*}
\Norm{\cT^n-\cT_\infty^n}{}\leq\beta^{n+1}\frac{4C e\sqrt{\abs{m}}}{(1-\beta)^{3/2}}\:\sup_{\abs{z}=\beta}\left|\frac{1}{B(z)}\right|.
\end{align*}
\end{corollary}

A detailed discussion of Theorem~\ref{result} and Corollary~\ref{corresult} follows in Section~\ref{discussion}. Here, let us just mention some situations in which the above bounds might be applied.
\begin{enumerate}
\item When $\cT$ is the transition matrix of a classical time-homogenous Markov chain, Theorem~\ref{result} can be used to estimate the distance of $\cT^n$ to stationarity. For the classical 1-to-1 norm it holds that $\Norm{\cT^n}{1\rightarrow1}=1$ for any stochastic matrix and any natural number $n$, such that Theorem~\ref{result} applies with $C=1$.
\item For all $\cT\in\mathfrak{T}_+$ and any $n$ we have that $\Norm{\cT^n}{\diamond}=1$. Thus, Theorem~\ref{result} provides a convergence bound for quantum Markov chains with $C=1$. 
\item Theorem~\ref{result} holds for general power bounded operators (in finite dimensions) whose spectrum is contained in the unit disc. Therefore our result applies to cone- and base-preserving maps with the corresponding norms, more general than transition matrices of classical Markov chains and TPPMs. An important class of such operations constitute LOCC maps \cite{projmet}.
\item In the context of classical and quantum Markov chains one is often interested in the quantity $\Norm{\cT^n(v)-\cT_\infty^n(v)}{1}$, where, depending on the context, $v$ is either a probability vector or a quantum state. If $v$ is contained in an invariant subspace $\cV_{inv}$ of $\cT$ it is clear that one can improve the bound in Theorem~\ref{result}. We then have that
\begin{align*}
\Norm{\cT^n(\rho)-\cT_\infty^n(\rho)}{1}
&\leq\frac{4e^{2}\sqrt{\abs{m}}(\abs{m}+1)\mu^{n+1}}{n\left(1-(1+\frac{1}{n})\mu\right)^{3/2}}\:\sup_{\abs{z}=\mu(1+1/n)}\left|\frac{1}{B(z)}\right|\\
&\leq\frac{4 e^{2}\sqrt{\abs{m}}(\abs{m}+1)\mu^{n}}{\left(1-(1+\frac{1}{n})\mu\right)^{3/2}}\prod_{m/(z-\lambda_{D})}\frac{1-(1+\frac{1}{n})\mu\abs{\lambda_i}}{\mu-\abs{\lambda_i}+\frac{\mu}{n}},
\end{align*}
where now $B=B_{(\cT-\cT_\infty)_{inv}}$ is the Blaschke product corresponding to the minimal polynomial $m=m_{(\cT-\cT_\infty)_{inv}}$ of $\cT-\cT_\infty$ restricted to $\cV_{inv}$.
\item 
If $\cT$ has a unique eigenvalue of magnitude one Corollary~\ref{geometric} establishes a geometric estimate for the convergence towards the stationary state of the chain. In \cite{wir} this is used to analyze the sensitivity of the stationary states of such Markov chains to perturbations in the transition map. Even in the context of classical Markov chains stability estimates based on Corollary~\ref{geometric} yield a significant improvement, see~\cite{wir}. The core conceptual insight is that the (inverse pseudo-hyperbolic) distance of the eigenvalues of $\cT$ to $\beta$ determines the sensitivity of the chain to perturbation. This is in contrast to previous work \cite[Thm. 4.1]{Mit1}, where corresponding estimates involved (inverse) distances $\abs{\lambda_i-\lambda_j}^{-1}$, which diverge when the spectrum becomes degenerate. More generally, based on power-boundedness of the transition map one can prove strong spectral stability estimates~\cite{ich} and strengthen the estimates of e.g.~\cite{sense}.
\end{enumerate}

For the proof we present an upper bound on $\Norm{\cA^n}{}$ for a general power bounded operator $\cA$, whose spectrum is contained in $\mathbb{D}$ and we specialize to the case $\cA=\cT-\cT_\infty$ only at the end. More precisely, we start with any $\cA\in\linops{\cV}$ whose spectrum is contained in the open unit disc and suppose that $\Norm{\cA^n}{}\leq C$ for all $n\in\mathbb{N}$. We employ Lemma~\ref{wiener} to obtain an estimate in terms of $\Norm{z^n}{W/mW}$. The key point is then to find a good bound on $\Norm{z^n}{W/mW}$. Our approach to this problem is inspired by the proof of the sharp Kreiss matrix theorem~\cite{ttad,LeVeque,Spijker}. For convenience we shall assume that the eigenvalues $\{\lambda_1,...,\lambda_D\}$ of $\cA$ are ordered with non-decreasing magnitude and that the map $\cA$ is diagonalizable, i.e., its minimal polynomial decomposes into pairwise distinct linear factors. This assumption does not lead to any difficulties when it comes to finding upper bounds of the type of Theorem~\ref{result}. To see this, assume that, for each fixed $n$, Theorem~\ref{result} holds true for any $\cA$ such that the minimal polynomial $m_{\cA}$ decomposes into pairwise distinct linear factors. To pass to the case when $\cA$ has non-trivial Jordan structure one slightly perturbs the spectrum of $\cA$ and obtains a diagonalizable map $\cA+\epsilon$. Note that for sufficiently small $\epsilon$ the spectrum of $\cA+\epsilon$ still is contained in the open unit disc, such that $\cA+\epsilon$ is power bounded with some constant $C_\epsilon$. In the limit of $\epsilon\rightarrow0$, $C_\epsilon$ converges to $C$ \cite{AC}. Thus, for each fixed $n$ one can apply the theorem for diagonalizable matrices and pass to the limit $\epsilon\rightarrow 0$ on both sides of Theorem~\ref{result}. By continuity of the norm this implies the claimed statement.
\begin{proof}[Proof of Theorem~\ref{result}] We adapt techniques developed in \cite{AC} for general power bounded operators (see Theorem 3.20) and invoke Lemma~\ref{wiener} to transfer the problem of estimating $\Norm{\cA^n}{}$ to the one of bounding $\Norm{z^n}{W/mW}$. It follows from the definition of the function algebra, that \cite{AC}
\begin{align}
\Norm{z^n}{W/mW}=\inf\{\Norm{g}{W}\:|\:g\in W,\: g(\lambda_i)=\lambda_i^n\}\label{here}.
\end{align}
This means that the problem of bounding $\Norm{z^n}{W/mW}$ is equivalent to finding a minimal norm function $g$ that interpolates the data set $(\lambda_1,\lambda_1^n),...,(\lambda_{\abs{m}},\lambda_{\abs{m}}^n)$ in the sense that $g(\lambda_i)=\lambda_i^n$.
More generally, the task of bounding a function $f$ of a quantum channel is related to an interpolation problem in the Wiener algebra by replacing $\lambda_i^n$ by $f(\lambda_i)$. The strategy of our proof will be to consider one specific representative function $g$ in \eqref{here} and bound its norm. To achieve this we employ the following method. Instead of considering $g$ directly we choose a ``smoothing parameter\rq\rq{} $r$ and pass to a ``stretched\rq\rq{} interpolation function.\\
Given any function $f\in H^2$ and $r\in(0,1)$, we write $f_r(z):=f(rz)=\sum_{k\geq0}\hat{f}(k)r^kz^k$ and observe that by the Cauchy-Schwarz Inequality
\begin{align}
\Norm{f_r}{W}\leq\sqrt{\sum_{k\geq0}\abs{\hat{f}(k)}^2}\sqrt{\frac{1}{1-r^2}}\leq\Norm{f}{H^\infty}\sqrt{\frac{1}{1-r^2}}.\label{CS}
\end{align}
This idea was used to obtain bounds to the inverse and resolvent of a power bounded operator in \cite{AC}.\\ 
We use the Blaschke products $B(z)=\Pi_i\frac{z-\lambda_i}{1-\bar{\lambda}_i z}$ and $\tilde{B}(z)=\Pi_i\frac{z-r\lambda_i}{1-r\bar{\lambda}_i z}$, where in the latter product the spectrum is stretched by a factor of $r$. (The products are taken over all prime factors of $m_\cA$, but to avoid cumbersome notation we do not write this explicitly.) Consider now the function $g$ with
\begin{align*}
g(z)=\sum_k \lambda_k^n\frac{B(z)}{z-\lambda_k} (1-\abs{\lambda_k}^2)\prod_{j\neq k}\frac{1-\bar{\lambda}_j\lambda_k}{\lambda_k-\lambda_j}.
\end{align*}
$g$ is analytic in the unit disc and $g(\lambda)=\lambda^n$ for all $\lambda\in\sigma(T)$. To be able to use the estimate~\eqref{CS} we perform the aforementioned smoothing. We define the modified function $\tilde{g}$ by
\begin{align*}
\tilde{g}(z)=\sum_k \lambda_k^n\frac{\tilde{B}(z)}{z-r\lambda_k}(1-r^2\abs{\lambda_k}^2)\prod_{j\neq k}\frac{1-r^2\bar{\lambda}_j\lambda_k}{r\lambda_k-r\lambda_j}
\end{align*}
and observe that $\tilde{g}_r$ enjoys the same basic properties as $g$, i.e., $\tilde{g}_r$ is analytic in $\mathbb{D}$ and $\tilde{g}_r(\lambda)=\lambda^n$ for any $\lambda\in\sigma(T)$.
Thus, by Equation~\eqref{here}, we have that $\Norm{z^n}{W/mW}\leq\Norm{\tilde{g}_r}{W}$ and it follows from Inequality~\eqref{CS} that
\begin{align*}
\Norm{\tilde{g}_r}{W}\leq\sqrt{\frac{1}{1-r^2}}\Norm{\tilde{g}}{H^\infty}.
\end{align*}
By the Maximum Principle for analytic functions $\Norm{\tilde{g}}{H^\infty}$ is attained on the unit circle, that is {$\Norm{\tilde{g}}{H^\infty}=\sup_{\abs{z}=1}\abs{\tilde{g}(z)}$}. Exploiting the fact that each elementary Blaschke factor preserves the unit circle, we conclude that
\begin{align*}
\Norm{\tilde{g}}{H^\infty}=\sup_{\abs{z}=1}{\left|\sum_k\lambda_k^n \frac{1-r^2\abs{\lambda_k}^2}{z-r\lambda_k}\prod_{j\neq k}\frac{1-r^2\bar{\lambda}_j\lambda_k}{r\lambda_k-r\lambda_j}\right|}.
\end{align*}
To bound this quantity we perform a contour integration along the circle $\gamma: \phi\mapsto s e^{i\phi}$, where $s<1$ is chosen in a way such that $\gamma$ encircles all eigenvalues of $\cA $. By the Residue Theorem (note that $\abs{z}=1$) we have that
\begin{align}
\sum_k\lambda_k^n  \frac{1-r^2\abs{\lambda_k}^2}{z-r\lambda_k} \prod_{j\neq k}\frac{1-r^2\bar{\lambda}_j\lambda_k}{r\lambda_k-r\lambda_j}=
\frac{1}{2\pi i}\int_\gamma\frac{\lambda^n}{\tilde{B}_r(\lambda)}\frac{1}{z-r\lambda}\d{\lambda}.\label{integral}
\end{align}
Integration by parts gives 
\begin{align}
\frac{1}{2\pi i}\int_\gamma\frac{\lambda^n}{\tilde{B}_r(\lambda)}\frac{1}{z-r\lambda}\d{\lambda}=
&-\frac{1}{2\pi i (n+1)} \int_\gamma\lambda^{n+1}\left[\frac{1}{\tilde{B}_r(\lambda)(z-r\lambda)}\right]\rq{}\d{\lambda}\label{byparts}
\end{align}
and we arrive at
\begin{align*}
\Norm{\tilde{g}}{H^\infty}\leq\frac{ s^{n+1}}{2\pi(n+1)} \sup_{\abs{z}=1} \int_\gamma\left|\left[\frac{1}{\tilde{B}_r(\lambda)(z-r\lambda)}\right]\rq{}\right|\abs{\d{\lambda}}.
\end{align*}
The right hand integral can be interpreted as the arc length of the image of $\gamma$ under the rational function $\frac{1}{\tilde{B}_r(\lambda)(z-r\lambda)}$. For this quantity we have by Spijker\rq{}s Lemma (\cite{Spijker}, Equation (4))
\begin{align*}
\int_\gamma\left|\left[\frac{1}{\tilde{B}_r(\lambda)(z-r\lambda)}\right]\rq{}\right|\abs{\d{\lambda}}\leq
2\pi(\abs{m}+1)\sup_{\abs{\lambda}=s}\left|\frac{1}{\tilde{B}_r(\lambda)(z-r\lambda)}\right|
\end{align*}
and conclude that for $0<r<1$ and $\mu<s<1$ we have
\begin{align*}
\Norm{\tilde{g}}{H^\infty}\leq s^{n+1}\frac{(\abs{m}+1)}{(n+1)} \frac{1}{1-rs}\sup_{\abs{\lambda}=s}\left|\prod_i\frac{1-\bar{\lambda}_i r^2 \lambda}{r\lambda-r\lambda_i}\right|.
\end{align*}
In the above bound we choose $s=(1+1/n)\mu$ (where $\mu$ denotes the spectral radius of $\cA$) and notice that 
\begin{align*}
s^{n+1}=\mu^{n+1}\left(1+\frac{1}{n}\right)^{n+1}\leq e(1+1/n)\mu^{n+1},
\end{align*}
which entails
\begin{align*}
&\Norm{\tilde{g}}{H^\infty}
\leq\frac{\mu^{n+1}(\abs{m}+1)e}{nr^{\abs{m}}(1-r(1+1/n)\mu)}\sup_{\abs{\lambda}=\atop(1+1/n)\mu}\left|\prod_i\frac{1-\bar{\lambda}_i r^2 \lambda}{\lambda-\lambda_i}\right|\\
\end{align*}
and
\begin{align*}
\Norm{z^n}{W/mW}\leq
\sqrt{\frac{1}{1-r^2}}\frac{\mu^{n+1}(\abs{m}+1)e}{nr^{\abs{m}}(1-r(1+1/n)\mu)}\sup_{\abs{\lambda}=\atop(1+1/n)\mu}\left|\prod_i\frac{1-\bar{\lambda}_i r^2 \lambda}{\lambda-\lambda_i}\right|.
\end{align*}
Finally, we observe that
\begin{align*}
\sup_{\abs{\lambda}=\atop(1+1/n)\mu}\left|\prod_i\frac{1-\bar{\lambda}_i r^2 \lambda}{\lambda-\lambda_i}\right|
&=\sup_{\abs{\lambda}=\atop(1+1/n)\mu}\left|\frac{1}{B(\lambda)}\right|\ \cdot\  \prod_{i}\left|1+\frac{\bar{\lambda}_i\lambda(1-r^2)}{1-\bar{\lambda}_i\lambda}\right|\\
&\leq\sup_{\abs{\lambda}=\atop(1+1/n)\mu}\left|\frac{1}{B(\lambda)}\right|\ \cdot\ \left(1+\frac{1-r^2}{1-\mu(1+1/n)}\right)^{\abs{m}}.
\end{align*}
We can choose $1-r^2=\frac{1-\mu(1+1/n)}{\abs{m}}$ and get 
\begin{align*}
\Norm{z^n}{W/mW}\leq\frac{2e^{2}\mu^{n+1}\sqrt{\abs{m}}(\abs{m}+1)}{n(1-(1+1/n)\mu)^{3/2}}\sup_{\abs{\lambda}=\atop(1+1/n)\mu}\left|\frac{1}{B(\lambda)}\right|,
\end{align*}
where we used the fact that $(1+1/\abs{m})^{\abs{m}}\leq e$ and that, by the Bernoulli inequality for $\abs{m}>1$, $r^{\abs{m}}\geq(1-\frac{1-\mu(1+\frac{1}{n})}{2})\geq1/2$.\\
We now specialize the above derivation to the case when $\cA=\cT-\cT_\infty$. By assumption it holds for any $n$ and $\cT\in\mathfrak{T}$ that $\Norm{\cT^n}{}\leq C$ and it follows that
\begin{align*}
\Norm{(\cT-\cT_\infty)^n}{}=\Norm{\cT^n-\cT_\infty^n}{}\leq\Norm{\cT^n}{}+\Norm{\cT_\infty^n}{}\leq 2C.
\end{align*}
In total we can assert that
\begin{align*}
\Norm{\cT^n-\cT^n_\infty}{}=\Norm{(\cT-\cT_\infty)^n}{}
&\leq \frac{4 C e^{2}\abs{m_{\cT-\cT_\infty}}^{1/2}(\abs{m_{\cT-\cT_\infty}}+1)\cdot \mu^{n+1}}{n\left(1-(1+\frac{1h}{n})\mu\right)^{3/2}}\sup_{\abs{\lambda}=\atop(1+1/n)\mu}\left|\frac{1}{B(\lambda)}\right|
\end{align*} 
completing the proof of Theorem~\ref{result}.
\end{proof}
\begin{proof}[Proof of Corollary~\ref{geometric}] The corollary results from a simple truncation of the proof of Theorem~\ref{result}. It follows from \eqref{integral} that

\begin{align*}
\Norm{\tilde{g}}{H^\infty}=\frac{1}{2\pi}\sup_{\abs{z}=1}\left|\int_\gamma\frac{\lambda^n}{\tilde{B}_r(\lambda)}\frac{1}{z-r\lambda}{\d\lambda}\right|
\leq\frac{s^{n+1}}{1-rs}\sup_{\abs{\lambda}=s}\left|\frac{1}{\tilde{B}_r(\lambda)}\right|,
\end{align*}
where $s\in(\mu,1)$. One can bound as before
\begin{align*}
\left|\frac{1}{\tilde{B}_r(\lambda)}\right|\leq\frac{1}{r^{\abs{m}}}\left(1+\frac{1-r^2}{1-s}\right)^{\abs{m}}\frac{1}{\abs{B(\lambda)}}
\end{align*}
and choose $1-r^2=\frac{1-s}{\abs{m}}$. It follows
\begin{align*}
\Norm{z^n}{W/mW}\leq2es^{n+1}\frac{\sqrt{\abs{m}}}{(1-s)^{3/2}}\sup_{\abs{\lambda}=s}\frac{1}{\abs{B(\lambda)}}.
\end{align*}
\end{proof}
\subsection{Comparison to the Schur and Jordan convergence bounds}\label{discussion}
Theorem~\ref{result} significantly improves upon both the Jordan and the Schur bounds, Theorems~\ref{schur},~\ref{jordan}. In this subsection we shall illustrate this comparing the different convergence estimates for a semigroup of quantum channels. Since for all $\cT\in\mathfrak{T_+}$ we have that $\Norm{\cT}{\diamond}=1$, Theorem~\ref{result} gives a bound for the diamond norm. With the notation of Theorem~\ref{result} we have that
\begin{align}
\Norm{\cT^n-\cT_\infty^n}{\diamond}&\leq\frac{4 e^{2}\sqrt{\abs{m}}(\abs{m}+1)\cdot \mu^{n+1}}{n\left(1-(1+\frac{1}{n})\mu\right)^{3/2}}\:\sup_{\abs{z}=\atop\mu(1+1/n)}\left|\frac{1}{B(z)}\right|\label{szrewothm}\\
&\leq\frac{4 e^{2}\sqrt{\abs{m}}(\abs{m}+1)\cdot \mu^{n}}{\left(1-(1+\frac{1}{n})\mu\right)^{3/2}}\:\prod_{m/(z-\lambda_{D})}\frac{1-(1+\frac{1}{n})\mu\abs{\lambda_i}}{\mu-\abs{\lambda_i}+\frac{\mu}{n}}.\label{szrewo}
\end{align}
%
%
For the inverse Blaschke product in \eqref{szrewo} we can establish lower and upper bounds. The function $\frac{1-(1+1/n)\mu x}{(1+1/n)\mu-x}$ is monotonically increasing with $x\in[0,\mu]$ and we have that
\begin{align}
\left(\frac{1}{(1+1/n)\mu}\right)^{\abs{m}-1}\leq
\prod_{m/(z-\lambda_{D})}\frac{1-(1+\frac{1}{n})\mu\abs{\lambda_i}}{\mu-\abs{\lambda_i}+\frac{\mu}{n}}\leq\ \left(\frac{n}{\mu}(1-\mu^2)\right)^{\abs{m}-1}.\label{compare}
\end{align}
In the following we compare Inequalities~\eqref{szrewothm},~\eqref{szrewo} to the corresponding bounds resulting from the Jordan and Schur decompositions.\\
\subparagraph*{Comparison with the Jordan bound.} To establish a convergence bound for quantum channels in diamond norm one can use Theorem~\ref{jordan} together with the norm equivalence relations~\eqref{diamond1to1},~\eqref{equival}. But as Theorem~\ref{jordan} has a qualitative character only (i.e., it does not specify $C_1,\ C_2$), the constants coming from the norm equivalence are of no relevance. As expected, both Theorem~\ref{jordan} and Inequality~\eqref{szrewo} include an exponential factor $\mu^n$. Suppose that the largest Jordan block for $\lambda_D$ has size $d_\mu$ and that there is no other eigenvalue of $\cT-\cT_\infty$ of magnitude $\mu$. Then the minimal polynomial of $\cT-\cT_\infty$ contains a factor $(z-\lambda_D)^{d_\mu}$ and in~\eqref{szrewo} there are $d_\mu-1$ factors for this eigenvalue. The denominator in Inequality~\eqref{szrewo} leads to a factor $(n/\mu)^{d_\mu-1}$ in this estimate, which captures the same qualitative $n$-dependence as the upper bound of Theorem~\ref{jordan}. Due to the lower bound in Theorem~\ref{jordan} the factor $(n/\mu)^{d_\mu-1}$ is also necessary. But as compared to Theorem~\ref{jordan} Inequality~\eqref{szrewo} bears the obvious advantage that it specifies  $C_2$.
On the other hand if there are several distinct eigenvalues of magnitude $\mu$, Inequality~\eqref{szrewo} does not yield the correct asymptotic behavior from Theorem~\ref{jordan}, since \emph{any} eigenvalue of magnitude $\mu$ occurring in $m$ contributes a factor $n/\mu$. The reason for this lies in the estimate~\eqref{szrewothm}, i.e., in bounding each Blaschke factor individually, which leads to Corollary~\ref{corresult}. Roughly speaking, if there are distinct eigenvalues of magnitude $\mu$ then, for sufficiently large $n$, any $z$ of magnitude $\mu(1+1/n)$ can be close at most to one of those eigenvalues. It is not difficult to make this intuition precise and prove the upper bound of Theorem~\ref{jordan} based on Theorem~\ref{result} with the additional advantage of specifying $C_2$. Finally we note that the occurrence of the correct asymptotic $n$-dependence in Theorem~\ref{result} is linked to the integration by parts in \eqref{byparts} and our application of Spijker\rq{}s Lemma. This procedure yields the $1/n$ prefactor in Theorem~\ref{result}, which is canceled by one inverse Blaschke factor in Corollary~\ref{corresult}. Had we bounded \eqref{integral} directly by the supremum of the integrand on the circle, we would have obtained an estimate where one factor in the Blaschke product is proportional to $n/\mu$ even in case of only one eigenvalue of magnitude $\mu$.\\
\subparagraph*{Comparison with the Schur bound.}  Taking into account the norm equivalence relations~\eqref{diamond1to1},~\eqref{equival} the Schur bound entails

\begin{align*}
\Norm{\cT^n-\cT_\infty^n}{\diamond}\leq 2D^{3/4}(\mu+2D^{1/4})^{D-1}n^{D-1}\mu^{n-D+1}.
\end{align*}
If one assumes that $\lambda_{D}$ is $D$-fold degenerate with maximal Jordan block this results in a factor $(n/\mu)^{D-1}$ in Inequality \eqref{compare}. Hence, even in the case of the worst Jordan structure for $\cT-\cT_\infty$, Theorem~\ref{result} improves upon bounds obtained from Theorem~\ref{schur} \emph{exponentially} in the $D$-dependent prefactor.

Finally, we discuss some implications of the lower bound in \eqref{compare}. We use that bound to estimate how good the upper bound of Corollary~\ref{corresult} can possibly be. Note that the left hand side of Inequality~\eqref{compare} contains a factor $(1/\mu)^{\abs{m}-1}$. If all eigenvalues of $\cT$ are distinct this factor grows with the dimension of the system. That is, for \lq\lq{}generic\rq{}\rq{} $\cT$ it needs $D$ time steps until Corollary~\ref{corresult} can yield a nontrivial statement. This is unfortunate from the point of view of applications, where one is looking for estimates such that $poly(\log(D))$ steps are sufficient. It is natural to ask whether or not Theorem~\ref{result} is optimal and whether one might be able to dispense of the $(1/\mu)^{\abs{m}-1}$ prefactor. The following subsection discusses aspects related to the optimality of Theorem~\ref{result}. Even full information about the spectrum (alone) is never sufficient to prove $poly(\log{D})$ convergence. To overcome this issue one may use properties of the semigroup beyond its spectrum. One important class of semigroups for which fast convergence can be proved under additional assumptions are detailed balanced semigroups (Definition~\ref{detailedbalancedef}). We discuss the convergence of such semigroups in detail in Section~\ref{secdetbal}.

\subsection{Semigroups of Hilbert space contractions} \label{contractive} 
In this subsection we discuss semigroups of Hilbert space contractions. More precisely, suppose we are given a semigroup $(\cT^n)_{n\geq0}$ of linear operators acting on a finite-dimensional Hilbert space such that $\Norm{\cT}{\infty}\leq1$. As before, our major interest lies in bounding the quantity $\Norm{\cT^n-\cT_\infty^n}{\infty}$ in terms of the spectrum of $\cT$. Clearly, this setup is less general than our main setup in Section~\ref{semi} and one can expect better bounds. In what follows we derive an analog of Theorem~\ref{result} for contractive semigroups and discuss the optimality of the obtained bounds.

Let us adopt the notation from Theorem~\ref{result}. As before we write $\sigma{(\cT-\cT_\infty)}$ for the spectrum and $m=m_{\cT-\cT_\infty}$ for the minimal polynomial of $\cT-\cT_\infty$. $B(z)=\prod_i\frac{z-\lambda_i}{1-\bar{\lambda}_iz}$ denotes the Blaschke product associated with $m$. To avoid cumbersome notation we shall again assume that $m$ has simple zeros. The extension to the more general case does not result in any difficulties.
Before we proceed with our main discussion we briefly introduce some notation and standard concepts from spectral operator theory. We define the $\abs{m}$-dimensional model space
\begin{align*}
K_B:=H^2\ominus BH^2:=H^2\cap(BH^2)^\bot,
\end{align*}
where we employ the usual scalar product from the Hilbert space $H^2$. The model operator $M_B$ acts on $K_B$ as
\begin{align*}M_B:\:K_B&\rightarrow K_B\\
f&\mapsto M_B(f)=P_B(zf),
\end{align*}
where $P_B$ denotes the orthogonal projection on $K_B$. In other words, $M_B$ is the compression of the multiplication operation by $z$ to the model space $K_B$ (see \cite{TS} for a detailed discussion of model operators and spaces). As multiplication by $z$ has operator norm $1$ it is clear that $M_B$ is a Hilbert space contraction. More precisely, for any $\phi\in H^\infty$ the norm of $\phi(M_B)$ can be evaluated using Sarason\rq{}s approach to interpolation theory~\cite{Sarason,OF} as 
\begin{align}
\Norm{\phi(M_B)}{\infty}=\Norm{\phi}{H^\infty/mH^\infty}.\label{sarasonnorm}
\end{align}
We can also write $\Norm{\phi(M_B)}{\infty}$ as variational expression in the Hardy space $H^1$. From \cite{Ext} we get that
\begin{align}
\Norm{\phi(M_B)}{\infty}=\sup_{F\in H^1\atop\Norm{F}{1}\leq1}\left|\frac{1}{2\pi i}\int_{\abs{z}=1}\frac{\phi}{B}F\:\d z\right|.\label{extremal}
\end{align}
Note that this trivially implies
\begin{align*}
\left|\frac{1}{2\pi i}\int_{\abs{z}=1}\frac{\phi}{B}\:\d z\right|\leq\Norm{\phi(M_B)}{\infty}\leq\sup_{\abs{z}=1}\left|\frac{\phi}{B}\right|.
\end{align*}
It can be shown that the spectrum of the model operator $M_B$ defined above is given by the zeros of the corresponding Blaschke product $B$. In our case this means that $\cT-\cT_\infty$ and $M_B$ have identical spectrum. Hence, to any $\cT$ we can associate a (completely non-unitary \cite{Nagy}) contraction $M_B$ having spectrum $\sigma=\sigma{(\cT-\cT_\infty)}$.

Let us proceed by studying convergence estimates for the contractive semigroup of the form of Inequality~\eqref{firstbound}.
To start with, we prove that if $\Norm{\cT}{\infty}\leq1$ then {$\Norm{\cT-\cT_\infty}{\infty}\leq1$}, i.e. the semigroup $\{(\cT-\cT_\infty)^n\}_{n\geq0}$ is contractive, too.
\begin{proposition}\label{contra}
Let $(\cT^n)_{n\geq0}$ be a contractive semigroup on a Hilbert space and let $\cT_\infty$ be as in Equation~\eqref{Tinfty}. Then\\ 
(i) the semigroup $\{(\cT-\cT_\infty)^n\}_{n\geq0}$ is contractive, and\\
(ii) if $\cT^*(e)=\lambda e$ with $\abs{\lambda}=1$, then $\cT(e)=\bar{\lambda} e$.
\end{proposition}
\begin{proof}
Both follows from the fact that any contraction on a Hilbert space admits a unique decomposition into an orthogonal direct sum of a unitary and a completely non-unitary operation (\cite{Nagy}, Theorem 3.2). In our case, $\cT_\infty$ corresponds exactly to the unitary part of $\cT$ and $\cT-\cT_\infty$ is a (completely non-unitary) contraction, hence (i). (ii) is then a consequence of the normality of the unitary part. 
\end{proof}
The second part of Proposition~\ref{contra} generalizes the fact that for classical as well as for quantum Markov processes, contractivity implies that the transition map is doubly stochastic. In fact, in those cases the converse implication holds as well \cite{RuWo}.

From the first part  of Proposition~\ref{contra} and by Inequality~\eqref{calc1} it follows that
\begin{align*}
\Norm{\cT^n-\cT_\infty^n}{\infty}=\Norm{(\cT-\cT_\infty)^n}{\infty}\leq\Norm{z^n}{H^\infty}.
\end{align*}
Our previous considerations from Section~\ref{ff}, Lemma~\ref{wiener} furthermore imply
\begin{align}
\Norm{\cT^n-\cT_\infty^n}{\infty}\leq\Norm{z^n}{H^\infty/mH^\infty}.\label{Hinfcont}
\end{align}
We conclude from Equation~\eqref{sarasonnorm} that in order to upper bound~\eqref{Hinfcont} it is sufficient to consider $\Norm{M_B^n}{\infty}$. This is in contrast to our discussion of bounded semigroups on Banach spaces, where we had to rely on the Cauchy-Schwarz Inequality~\eqref{CS}. In addition, we note that $\Norm{\phi}{H^\infty/mH^\infty}=\Norm{\phi(M_B)}{\infty}$ allows us to work with $\Norm{\cdot}{H^\infty/mH^\infty}$ directly and we do not require an ad hoc function to upper bound \eqref{here}.

In our study of  bounded semigroups in Section~\ref{discussion} we have encountered a factor $(1/\mu)^{\abs{m}-1}$ in \eqref{compare} that grows exponentially with the dimension of the space on which the semigroup acts if all eigenvalues of the generator are distinct. The following proposition shows that, if in a bound of the type~\eqref{firstbound} $K$ only depends on the eigenvalue structure of $\cT$ and on $n$, then $K$ must contain such a factor. We achieve this by showing that for any contractive semigroup with generator $\cT$ there is a contractive semigroup whose generator has the same spectrum as $\cT$ but which converges slowly if $n$ is small.
\begin{proposition} \label{equalone}Let $(\cT^n)_{n\geq0}$ be a contractive semigroup acting on a $D$-dimensional Hilbert space and let $m=m_{\cT-\cT_\infty}$ denote the minimal polynomial of $\cT-\cT_\infty$ and $B$ the corresponding Blaschke product. Then there is a contractive semigroup $(\cE^n)_{n\geq0}$ such that $\cE$ has the \emph{same minimal polynomial} as $\cT$ and 
\begin{align*}
\Norm{\cE^n-\cE_\infty^n}{\infty}=\sup_{F\in H^1\atop\Norm{F}{1}\leq1}\left|\frac{1}{2\pi i}\int_{\abs{z}=1}\frac{z^n}{B} F\:\d z\right|.
\end{align*}
In particular, for all
$n<\abs{m}\leq D$
it holds that
\begin{align*}
\Norm{\cE^n-\cE_\infty^n}{\infty}=1.
\end{align*}

\end{proposition} 
The supremum in Proposition~\ref{equalone} is attained by a function $\tilde{F}=f^2$, where $f$ is in the unit ball of $K_B$ \cite{Ext}. Hence, the optimization effectively goes over a finite-dimensional vector space of rational functions with fixed poles and bounded degree (see \cite{Ext} for details). One can obtain simple lower bounds on the convergence speed of $(\cE^n)_{n\geq0}$ by choosing a certain $f\in K_B$ and evaluating the integral with the Residue theorem.\\

The second assertion of Proposition~\ref{equalone} states that for any spectrum we can construct a semigroup such that the distance of the evolution to its asymptotic behavior stays maximal for at least $\abs{m}-1$ time steps. Clearly this implies that one cannot prove that $poly(\log{\abs{m}})$ time steps bring the semigroup close to its stationary behavior if only spectral data is given. 

Note that if in a bound of the form $\Norm{\cT^n-\cT_\infty^n}{}\leq K\mu^n$, with a bounded semigroup $(\cT^n)_{n\geq0}$, $K$ only depends on the spectrum of $\cT$ then  by Proposition~\ref{equalone} we have $1\leq K\mu^{\abs{m}-1}$. That is, in this case we obtain the lower bound $K\geq(1/\mu)^{\abs{m}-1}$.

\begin{proof}[Proof of Proposition~\ref{equalone}] The first assertion is clear by choosing \lq\lq{}$\cE:=\cT_\infty\oplus M_B$\rq{}\rq{} such that $\cE_\infty=\cT_\infty$ (on the unitary subspace) and  $\cE-\cE_\infty=0\oplus M_B$. For the second we consider the extremal problem Equation~\eqref{extremal}. Let $\psi$ be any rational function with poles away from the unit circle $\abs{z}=1$. Corollary~5 in \cite{Ext} asserts that, we have 
\begin{align*}
\sup_{F\in H^1\atop\Norm{F}{1}\leq1}\left|\frac{1}{2\pi i}\int_{\abs{z}=1}\psi F\:\d z\right|=\sup_{\abs{z}=1}\left|\psi{(z)}\right|
\end{align*}
if and only if $\psi$ is a constant multiple of the quotient of two finite Blaschke products $B_1,B_2$ having no common zeros and such that the degree of $B_1$ is strictly smaller than the degree of $B_2$ ($\abs{B_1}<\abs{B_2}$), i.e., $\psi=c\frac{B_1}{B_2}$ for some $c\in\mathbb{C}$. Let $B$ denote the Blaschke product associated with $m$, it follows readily that 
\begin{align*}
\Norm{M_B^n}{\infty}=\sup_{F\in H^1\atop\Norm{F}{1}\leq1}\left|\frac{1}{2\pi i}\int_{\abs{z}=1}\frac{z^n}{B} F\:\d z\right|=1
\end{align*}
holds for $n<\abs{m}$.
\end{proof}

To gain a better understanding of weather the derivation of Theorem~\ref{result} is optimal, i.e. whether or not the obtained estimate is sharp, let us prove an analog of Theorem~\ref{result} for semigroups of Hilbert space contractions. The derivation is based on techniques similar to those that led to Theorem~\ref{result}, but in the case at hand we can take a more direct approach based on the theory of model operators. 
\begin{proposition}\label{contractiveresult} Let $(\cT^n)_{n\geq0}$ be a contractive semigroup on a $D$-dimensional Hilbert space and let $\cT_\infty$ be the operator introduced in \eqref{Tinfty} (i.e., the unitary part of $\cT$). We write $m=m_{\cT-\cT_\infty}$ for the minimal polynomial and $\mu$ for the spectral radius of $\cT-\cT_\infty$. $B$ denotes the Blaschke product associated with $m$. Then, for $n>\frac{\mu}{1-\mu}$ we have

\begin{align*}
\Norm{\cT^n-\cT_\infty^n}{\infty}\leq\mu^{n+1}\frac{2\abs{m} e}{n(1-(1+1/n)^2\mu^2)}\sup_{\abs{z}=\mu(1+1/n)}\left|\frac{1}{B(z)}\right|.
\end{align*}
\end{proposition}
As before, we can bound all terms in the Blaschke product individually (see Appendix~\ref{Blaschke}) and find (compare Corollary~\ref{corresult})
\begin{align*}
\Norm{\cT^n-\cT_\infty^n}{\infty}\leq\mu^n\frac{2\abs{m} e}{1-(1+1/n)^2\mu^2}\prod_{i\neq\abs{m}}\frac{1-\abs{\lambda_i}\mu(1+1/n)}{\mu(1+1/n)-\lambda_i}.
\end{align*}

\begin{proof}[Proof of Proposition~\ref{contractiveresult}] The derivation proceeds along the lines of Theorem~\ref{result}. We use an $H^\infty$ functional calculus to bound $\Norm{\cT^n-\cT_\infty^n}{\infty}$ in terms of $\Norm{z^n}{H^\infty /mH^\infty}$. The latter expression can be rewritten using a contour integral similar to Equation~\eqref{integral}, integrate by parts, and finally apply Spijker\rq{}s Lemma.  We have already mentioned that
\begin{align*}
\Norm{\cT^n-\cT_\infty^n}{\infty}&\leq\Norm{z^n}{H^\infty /mH^\infty}
=\Norm{M_B^n}{\infty}=
\sup_{F\in H^1\atop\Norm{F}{1}\leq1}\left|\frac{1}{2\pi i}\int_{\abs{z}=1}\frac{z^n}{B}F\:\d z\right|
\end{align*}
and that the supremum in this extremal problem is attained by some function $\tilde{F}=f^2$ with $f\in K_B$ \cite{Ext}. Thus, $\tilde{F}/B$ is a rational function with $2\abs{m}$ poles located at $(\xi_1,...,\xi_{\abs{m}},\bar{\xi}_1^{-1},....,\bar{\xi}_{\abs{m}}^{-1})$, where $\xi_i$ are the zeros of $m$. In the above integral we can change the contour of integration and integrate along the circle $\gamma:\phi\mapsto \mu(1+1/n)e^{i\phi}$. Integrating by parts and and applying Spijker\rq{}s Lemma \cite{Spijker} we obtain
\begin{align*}
\left|\frac{1}{2\pi i}\int_{\gamma}\frac{z^n}{B}\tilde{F}\:\d z\right|&=\frac{1}{2\pi (n+1)}\left|\int_{\gamma}z^{n+1}\left(\frac{\tilde{F}}{B}\right)\rq{}\:\d z\right|\\
&\leq\frac{\mu^{n+1}(1+1/n)^{n+1}}{2\pi(n+1)}\int_{\gamma}\left|\left(\frac{\tilde{F}}{B}\right)\rq{}\right|\abs{\d z}\\
&\ \leq\frac{2\abs{m}\mu^{n+1}(1+1/n)^{n+1}}{n+1}\sup_{\abs{z}=\mu(1+1/n)}\left|\frac{\tilde{F}}{B}\right|.
\end{align*}
It is  known that for $F\in H^1$ and $z\in\mathbb{D}$ one can bound $\abs{F(z)}\leq\frac{1}{1-\abs{z}^2}\Norm{F}{H^1}$ \cite{Macin} and with $(1+1/n)^n\leq e$ we finally obtain
\begin{align*}
\Norm{\cT^n-\cT_\infty^n}{\infty}&\leq\frac{2\abs{m} e\mu^{n+1}}{n(1-(1+1/n)^2\mu^2)}\sup_{\abs{z}=\mu(1+1/n)}\left|\frac{1}{B}\right|.
\end{align*}
%
\end{proof}
\subsection{Slow convergence for Markov chains} \label{pitty}
Proposition~\ref{equalone} provides an example of a slowly converging contractive semigroup with arbitrary given spectrum. One might wonder in how far the phenomenon extends to the Markov chain setup. When $\cT$ is the transition map of a classical or quantum Markov chain, is it possible to prove~\eqref{firstbound}
where, $K$ should only depend on the spectrum of $\cT$ and $n$ but such that the stationary behavior sets in after $poly(\log(D))$ time steps? The following example shows that this can not be the case.

We construct a classical stochastic $D\times D$ matrix $T$ with real positive spectrum such that $\Norm{T^n-T_\infty}{1\rightarrow 1}=2$ for $n\leq D-2$. Let, as always, $\mu$ denote the spectral radius of $T-T_\infty$. We write $\{e_i\}_{i=1,...,D}$ for the canonical column vectors, i.e., $(e_i)_j=\delta_{ij}$ and for $\lambda_i\in[0,1)$, $1\leq i\leq D-1$, we define
$$T:=\begin{pmatrix} \lambda_1 &  &  &  &  \\ 1-\lambda_1 & \lambda_2 & & & &  \\ & 1-\lambda_2 & \lambda_3& & &  \\ & & \ddots& \ddots& & \\ & & & & \lambda_{D-1}& \\  & & & & 1-\lambda_{D-1}& 1 \end{pmatrix}.$$
$T$ is a stochastic matrix with spectrum $\sigma(T)=\{\lambda_1,...,\lambda_{D-1},1\}$. Since $\lambda_i<1$ for large $n$ the image of $T^n$ converges to an one-dimensional subspace corresponding to the eigenvalue $1$. We have that $T_\infty=\lim_{n\rightarrow\infty}T^n$ and observe that $T_\infty e_1=e_D$.
It is not difficult to see that for $n\leq D-2$ the $D$-th entry of the vector $T^ne_1$ is always zero, $\braket{e_D}{T^ne_1}=0$. It follows that $\Norm{(T^n-T_\infty)e_1}{1}=2$ (where $\Norm{\cdot}{1}$ denotes the 1-norm, Section~\ref{clquMach}) and we conclude that $\Norm{T^n-T_\infty}{1\rightarrow 1}=2$ for $n\leq D-2$. As before, if $K$ only depends on the spectrum of $\cT$ this implies that $K\geq \left(1/\mu\right)^{D-2}$. Note that the above reasoning does not depend on the exact values of the eigenvalues (as long as they are non-negative). This suggests that generally the spectrum $\sigma(T)$ does not contain sufficient information to prove $poly(\log{D})$ fast convergence estimates. Since every classical stochastic matrix can be embedded into a quantum channel, the lower bound on $K$ is also true for quantum channels.

\section{Convergence bounds from detailed balance}\label{secdetbal}
Applications often rely on fast convergence in the sense that  $poly(\log{D})$ steps should suffice for the asymptotic behavior to set in.
In our previous discussion we have argued that such bounds cannot rely on spectral data alone. To obtain better convergence estimates one requires additional knowledge about the semigroup. In this section we will derive convergence estimates for a general bounded semigroup under the condition that its generator be related to a Hermitian map in a certain way -- for classical and quantum Markov processes this will correspond to the well-known \emph{detailed balance condition} (see, e.g., \cite{mixing,group,temme1}). Throughout this section we require the state space $\cV$ to be equipped with a scalar product $\braket{\cdot}{\cdot}$, which induces norms $\Norm{\cdot}{2}$ and $\Norm{\cdot}{\infty}$ on $\cV$ and $\linops{\cV}$, respectively, and for convenience we will sometimes assume an orthonormal basis in $\cV$ to be fixed (cf.~Subsection \ref{prel:not}).

\subsection{General bound}\label{subsectgeneraldetailedbalance}
We start with a generalization of the detailed balance condition for classical Markov chains. This allows us to employ the corresponding property in the context of bounded semigroups.
\begin{definition}[Detailed balance for linear maps]\label{detailedbalancedef}
Let a linear map $\cT\in\linops{\cV}$ be given. If $\cB\in\linops{\cV}$ is  positive-definite (i.e., $\braket{v}{\cB(v)}>0$~$\forall v\in\cV\setminus\{0\}$) and satisfies $\cT\cB=\cB\cT^*$, then we say that $\cT$ satisfies the detailed balanced condition(with respect to $\cB$).
\end{definition}
This definition is equivalent to saying that $\cT$ is Hermitian with respect to \emph{some} scalar product on the space $\cV$, namely the scalar product $\braket{\cdot}{\cB^{-1}(\cdot)}$, but we choose the formulation with \emph{given} scalar product $\braket{\cdot}{\cdot}$ (independent of $\cT$) and explicit use of $\cB$. Note further that, due to strict positive-definiteness, $\cB$ in the above definition is in particular Hermitian and invertible. In conventional formulations of the detailed balance condition the map $\cB$ is not required to be strictly positive-definite, but we do so here as the derived bounds become trivial otherwise (see below).

The detailed balance condition for a linear map $\cT$ gives
$$\cB^{-1/2}\cT\cB^{1/2}~=~\cB^{1/2}\cT^*\cB^{-1/2}~,$$
which means that $\cB^{-1/2}\cT\cB^{1/2}$ is Hermitian, therefore has only real eigenvalues $\lambda_i\in{\mathbb{R}}$ ($i=1,\ldots,D$), and is unitarily diagonalizable:
$$\cU^*\cB^{-1/2}\cT\cB^{1/2}\cU~=~\Lambda~=~\begin{pmatrix} \lambda_1 &  &  \\  & \ddots &  \\ &  & \lambda_D \end{pmatrix}~.$$
This equation implies that $\cT$ is diagonalized by the similarity transformation $S:=\cB^{1/2}\cU$ (i.e.~$S^{-1}\cT S=\Lambda$). Note that $\cT$ has spectrum $\{\lambda_i\}_i$, too.

If $\cT$ is now power-bounded, i.e.,~the generator of a bounded semigroup, the definition in Equation~\eqref{Tinfty} implies that $\cT_\infty$ is diagonalized by $S$ as well,
$$\cU^*\cB^{-1/2}\cT_\infty\cB^{1/2}\cU~=~\Lambda_\infty,$$
where $\Lambda_\infty$ is obtained from $\Lambda$ by deleting all entries of magnitude smaller than $1$. $\Lambda-\Lambda_\infty$ is thus diagonal with operator norm $\mu<1$, where $\mu$ is the spectral radius of $\cT-\cT_\infty$. We thus arrive at the following convergence estimate:
\begin{align*}
\Norm{(\cT-\cT_\infty)^n}{\infty}&=\Norm{\cB^{1/2}\cU(\Lambda-\Lambda_\infty)^n\cU^*\cB^{-1/2}}{\infty}\\
&\leq\Norm{\cB^{1/2}}{\infty}\Norm{\cU}{\infty}\Norm{(\Lambda-\Lambda_\infty)^n}{\infty}\Norm{\cU^*}{\infty}\Norm{\cB^{-1/2}}{\infty}\\
&=\mu^n\Norm{\cB^{1/2}}{\infty}\Norm{\cB^{-1/2}}{\infty}~.
\end{align*}
(The latter two factors may be recognized as the \emph{condition number} of $\cB^{1/2}$.) We formulate this as a theorem:
\begin{theorem}\label{detailedbalancethm}
Let $\cV$ be a (real or complex) vector space with scalar product, and $\cT\in\linops{\cV}$ be the generator of a bounded semigroup $(\cT^n)_{n\geq0}$, which satisfies detailed balanced w.r.t.~a positive-definite $\cB\in\linops{\cV}$. Denote by $\mu$ the spectral radius of $\cT-\cT_\infty$. Then, for any $n\in\mathbb{N}$,
$$\Norm{\cT^n-\cT_\infty^n}{\infty}~\leq~\mu^n\Norm{\cB^{1/2}}{\infty}\Norm{\cB^{-1/2}}{\infty}~,$$
where $\Norm{\cdot}{\infty}$ denotes the operator norm on $\linops{\cV}$.
\end{theorem}

We now discuss detailed balance more specifically for classical and quantum Markov chains. First observe that, if $e\in\cV$ is a fixed point of $\cT^*$, i.e.~$\cT^*(e)=e$, then $\pi:=\cB(e)$ satisfies
$$\cT(\pi)=\cT\cB(e)=\cB\cT^*(e)=\cB(e)=\pi~,$$
i.e.,~$\pi$ is fixed by the semigroup generator $\cT$. Conversely, if $\pi$ is a fixed point of $\cT$, then $e:=\cB^{-1}(\pi)$ is left invariant by $\cT^*$. For a classical Markov chain the generator satisfies $T^*(e)=e$ with $e=\sum_{i=1}^De_i=(1,...,1)$ and $\cT^*(\id)=\id$
holds for generators of quantum Markov chains (see Section~\ref{clquMach}). Thus, for classical and quantum Markov chains the detailed balance condition immediately yields a fixed point of the transition map.

In the theory of classical Markov chains, a stochastic matrix $T\in{\mathbb{R}}^{d\times d}$ is usually defined to be detailed balanced w.r.t.\ the probability distribution $\pi\in{\mathbb R}^d$ (i.e.\ $\pi_i\geq0$ and $\sum_i\pi_i=1$), if $T_{ji}\pi_i=T_{ij}\pi_j$ holds for all $i,j$ (see, e.g., \cite{mixing,group}). Defining a diagonal matrix $B$ with entries $B_{ii}:=\pi_i$, the latter condition can be written as $TB=BT^*$. If furthermore the fixed-point probability distribution $\pi$ has full support (i.e.~$\pi_i>0$ $\forall i$), then $T$ is detailed balanced w.r.t.\ $B$ in the sense of our Definition~\ref{detailedbalancedef}. ($\pi=Be$ will necessarily be a fixed point of $T$.) Due to normalization it holds that $\min_i\pi_i\leq1/d$. Using this and the norm equivalence~\eqref{clequival}, Theorem \ref{detailedbalancethm} yields the following well-known convergence estimate \cite{mixing,group} for the special case of a classical Markov chain that satisfies detailed balance w.r.t.~the distribution $\pi$:
\begin{align}
\Norm{T^n-T_\infty^n}{1-1}~\leq~\mu^n\sqrt{d}\sqrt{\frac{\max_i\pi_i}{\min_i\pi_i}}~\leq~\frac{\mu^n}{\min_i\pi_i}~.\label{classmarkovchainestimate}
\end{align}
This estimate may become trivial if detailed balance is defined without the full-support condition on $\pi$ as one may then have $\min_i\pi_i=0$. On the other hand, if one has a positive lower bound on $\min_i\pi_i$, Equation~\eqref{classmarkovchainestimate} may become a useful convergence estimate. This technique is frequently used for detailed balanced chains that have a (unique) full-rank probability distribution as fixed point, and where one can find a ``good'' lower bound on $\min_i\pi_i$ \cite{mixing,group,cutoff}.
Often the situation arises that the chain converges to a Gibbs state $\pi_i=e^{-\beta H_i}/Z$ at finite inverse temperature $\beta\in[0,\infty)$ with $Z:=\sum_ie^{-\beta H_i}$. An important class of Markov chains that obey the detailed balance condition are Metropolis Hastings Markov Chains \cite{mixing}.

There are different generalizations of the detailed balance condition to quantum Markov chains \cite{temme1}, which we, however, all capture by Definition \ref{detailedbalancedef}. Let us specialize to the quantum detailed balance condition that most immediately generalizes the classical condition from the previous paragraph to the non-commutative case in a symmetric way and that has been employed for proving convergence of quantum Markov chains before (e.g. \cite{metro}). Namely, given a positive trace-preserving map $\cT\in\mathfrak{T}$ acting on the set $\cM_d$ of $d\times d$ matrices, we consider the detailed balance condition induced by the map $\cB_\sigma(X):=\sqrt{\sigma}X\sqrt{\sigma}$, where $\sigma\in\cM_d$ is a density matrix of full rank. Again, due to trace-preservation, it is easy to see that if $\cT$ is detailed balanced w.r.t.\ $\cB_\sigma$, then $\sigma=\cB_\sigma(\ii)$ is a fixed point of $\cT$. This leads to the following convergence result for the quantum case:
\begin{corollary}\label{coroquantdetailed}
Let $\cT:\cM_d\to\cM_d$ be a positive trace-preserving map, and $\sigma\in\cM_d$ be a full-rank density matrix such that
$$\sqrt{\sigma}\cT^*(X)\sqrt{\sigma}~=~\cT(\sqrt{\sigma}X\sqrt{\sigma})\quad\forall X\in\cM_d~.$$
Denote by $\mu$ the spectral radius of $\cT-\cT_\infty$. Then, for any $n\in\mathbb{N}$,
\begin{align*}
\Norm{\cT^n-\cT_\infty^n}{1-1}~\leq~\mu^n\sqrt{d}\sqrt{\frac{\lambda_{max}(\sigma)}{\lambda_{min}(\sigma)}}~\leq~\frac{\mu^n}{\lambda_{min}(\sigma)}~,
\end{align*}
where $\lambda_{min}(\sigma)$ and $\lambda_{max}(\sigma)$ denote the minimal and maximal eigenvalues of $\sigma$, respectively.
If, in addition, $\sigma=e^{-\beta H}/\trace{e^{-\beta H}}$ is the Gibbs state at inverse temperature $\beta\in[0,\infty)$ of a bounded Hamiltonian $H\in\cM_d$, then
\begin{align*}
\Norm{\cT^n-\cT_\infty^n}{1-1}~\leq~\mu^n\,d\,e^{2\beta\Norm{H}{\infty}}~.
\end{align*}
\end{corollary}
\begin{proof}[Proof of Corollary \ref{coroquantdetailed}]
The conditions on $\cT$ imply that it is detailed balanced w.r.t.~the map $\cB_\sigma$ defined above. Computing $\Norm{\cB_\sigma^{1/2}}{\infty}=\sqrt{\lambda_{max}(\sigma)}$ and $\Norm{\cB_\sigma^{-1/2}}{\infty}=1/\sqrt{\lambda_{min}(\sigma)}$ and considering the norm equivalence \eqref{equival} and bounding $\lambda_{min}(\sigma)\leq1/d$ and $\lambda_{max}(\sigma)\leq1$, we get the first assertion from Theorem \ref{detailedbalancethm}. In case of a thermal state, the second assertion follows from
$$\lambda_{min}\left(\frac{e^{-\beta H}}{\trace{e^{-\beta H}}}\right)~\geq~\frac{e^{-\beta\Norm{H}{\infty}}}{\trace{e^{\beta\Norm{H}{\infty}\ii_d}}}~=~\frac{e^{-2\beta\Norm{H}{\infty}}}{d}~.$$
\end{proof}

This corollary provides a possible way for proving that a state preparation or algorithm is \emph{efficient} in the sense of computational complexity \cite{NielsenChuang}. More concretely, for each $N$, consider a system of $N$ particles (spins), each with finite Hilbert space dimension $s<\infty$, and a Hamiltonian $H_N$ on each system. In many physical situations the Hamiltonian will be bounded by some polynomial of the particle number, $\Norm{H_N}{\infty}\leq c_HN^k$; this occurs for example if $H_N=\sum_iH_{N,i}$ is a sum of $k$-local terms that are uniformly bounded  by $c_H$. Assume further that the thermal state $\sigma_N=e^{-\beta H_N}/\trace{e^{-\beta H_N}}$ at inverse temperature $\beta\in[0,\infty)$ is a fixed point of the positive trace-preserving map $\cT_N$, and that $\cT_N$ satisfies detailed balanced w.r.t.\ $\cB_{\sigma_N}$. This assumption may be fulfilled, e.g., by Gibbs dynamics in a Markov chain Monte Carlo algorithm \cite{metro}. Lastly, assume that the \emph{spectral gap} of $\cT_N$ is asymptotically lower bounded by an inverse polynomial $c_\mu/N^\alpha$ of $N$ (where $c_\mu>0$), i.e., the eigenvalue $1$ corresponding to $\sigma_N$ is the only eigenvalue of $\cT_N$ with modulus $1$ whereas $|\lambda_i|\leq1-c_\mu/N^\alpha$ for all other eigenvalues. Among these assumptions, when they apply, the latter one is usually the hardest to prove in a given situation.

Under these presuppositions, the evolution operator $\cT_N$ prepares the final state $\sigma_N$ efficiently in the system size $N$. More precisely, for any initial state $\rho_N$ of the system, the time-evolved state $\cT_N^n(\rho_N)$ after $n$ steps will be $\varepsilon$-close in trace-norm to the thermal state $\sigma_N$ (i.e.~$\Norm{\cT_N^n(\rho_N)-\sigma_N}{1}\leq\varepsilon$) if
\begin{align}
n~\geq~\frac{N^\alpha}{c_\mu}\left(2\beta c_H N^k+N\log s+\log\frac{1}{\varepsilon}\right)~.\label{convergencetimequantumdetailed}
\end{align}
This means that the runtime to $\varepsilon$-convergence scales at most polynomially in the particle number $N$ and polylogarithmically in the desired accuracy $\varepsilon\in(0,1]$, which proves efficient state preparation.

For a proof of the runtime bound Inequality~\eqref{convergencetimequantumdetailed}, note that the dimension of the $N$-partite system is $d_N=s^N$ and that, due to the spectral gap condition, $(\cT_N)_\infty(\rho_N)=\sigma_N$ for any state $\rho_N$, which implies $\Norm{\cT_N^n(\rho_N)-\sigma_N}{1}\leq\Norm{\cT_N^n-(\cT_N)_\infty^n}{1-1}$. Finally, we use
$$\mu~\leq~1-\frac{c_\mu}{N^\alpha}~\leq~e^{-c_\mu/N^\alpha}$$
in the Gibbs state bound from Corollary \ref{coroquantdetailed} and requiring the latter to be at most $\varepsilon$ shows that the condition in Inequality~\eqref{convergencetimequantumdetailed} is sufficient for $\varepsilon$-convergence.

If one wants to bound the diamond norm $\Norm{\cT^n-\cT_\infty^n}{\diamond}$ between the actual and the asymptotic evolution in Corollary \ref{coroquantdetailed} instead of the trace-norm, then by Inequality~\eqref{diamond1to1} one incurs another factor $d$ (or $1/\lambda_{min}(\sigma)$) in the upper bounds. This however does not affect the efficiency statement just obtained, as the asymptotic dynamics $(\cT_N)_\infty$ is still reached, up to $\varepsilon$, in polynomial time.

\subsection{An $\ell^2$ bound}\label{llsec}
In this Subsection, again based on the detailed balance condition, we derive a sharper convergence bound than in Subsection \ref{subsectgeneraldetailedbalance}, taking into account all eigenvalues and eigenvectors of the transition map $\cT$. The special case of this bound for classical Markov processes has been used to prove so-called cutoff dynamics \cite{cutoff,group,mixing}. After describing the approach for general bounded semigroups obeying detailed balance, we will specialize to quantum Markov chains.

Recall from above that, if $\cT\in\linops{\cV}$ is detailed balanced w.r.t.~$\cB$, its eigenvalues $\lambda_i$ are real. Furthermore, as $\cR:=\cB^{-1/2}\cT\cB^{1/2}$ is a Hermitian operator, it has a complete orthonormal eigenbasis $\{x_i\}_i$, i.e.~$\cR(x_i)=\lambda_ix_i$. From this we can define an eigensystem of the adjoint $\cT^*$, which will play a prominent role in the bound:
$$y_i~:=~\cB^{-1/2}(x_i)\,,\quad\text{which implies~}\cT^*(y_i)~=~\lambda_i y_i\,.$$
$\{y_i\}_i$ could alternatively be chosen as any eigensystem of $\cT^*$ that is orthonormal w.r.t.~the weighted scalar product $\braket{\cdot}{\cB(\cdot)}$.

The spectral decomposition $\cR(v)=\sum_i\lambda_i \braket{x_i}{v} x_i$ now gives:
\begin{align*}
\cB^{-1/2}\cT^n(v)&=~\cR^n\cB^{-1/2}(v)\\
&=~\sum_{i=1}^D\lambda_i^n\braket{x_i}{\cB^{-1/2}(v)}x_i\\
&=~\sum_{i=1}^D\lambda_i^n\braket{\cB^{-1/2}(x_i)}{v}x_i~=~\sum_{i=1}^D\lambda_i^n\braket{y_i}{v}x_i~.
\end{align*}
Recognizing that $\cB^{1/2}(x_i)$ is the right-eigenvector of $\cT$ corresponding to $y_i$, the terms with $|\lambda_i|=1$ in the last expression (which we assume to be $i=r+1,\ldots,n$) correspond to the asymptotic evolution $\cT_\infty^n$. We can thus write
\begin{align*}
\cB^{-1/2}(\cT^n-\cT_\infty^n)(v)~=~\sum_{i=1}^r\lambda_i^n\braket{y_i}{v}x_i~,
\end{align*}
which, together with the fact that $\{x_i\}$ is an orthonormal system, gives by squaring:
\begin{align}
\Norm{(\cT^n-\cT_\infty^n)(v)}{2,\cB^{-1}}^2&:=~\braket{(\cT^n-\cT_\infty^n)(v)}{\cB^{-1}\left(\cT^n-\cT_\infty^n\right)(v)}
=~\sum_{i=1}^{r}\lambda_i^{2n}|\braket{y_i}{v}|^2~.\label{inconvl2bound}
\end{align}
This equality relates the eigensystem corresponding to the eigenvalues with modulus smaller than $1$ to the convergence in a suitably modified Hilbert norm. By itself this relation does not seem very useful, although one can derive Theorem \ref{detailedbalancethm} from it by rescaling the modified scalar product back to the originally given one.

When specializing to the quantum case, however, we can make a connection to the induced trace-norm, and thereby strengthen Corollary \ref{coroquantdetailed}:
\begin{proposition}\label{quantdetbalwitheigenvaluesprop}
Let $\cT:\cM_d\to\cM_d$ be a positive trace-preserving map, and $\sigma\in\cM_d$ be a full-rank density matrix (i.e.~$\trace{\sigma}=1$, $\sigma>0$) such that the detailed balance condition
$$\sqrt{\sigma}\cT^*(X)\sqrt{\sigma}~=~\cT(\sqrt{\sigma}X\sqrt{\sigma})\quad\forall X\in\cM_d$$
holds. Let $\{\lambda_i\}_{i=1}^r$ be the part of the spectrum of $\cT$ in the open interval $(-1,1)$, and $Y_i$ be the corresponding eigenvectors of the adjoint map $\cT^*$, orthonormal in the sense that $\trace{Y_i^*\sigma^{1/2}Y_j\sigma^{1/2}}=\delta_{ij}$. Then, for every $Z\in\cM_d$ (e.g.~a quantum state):
\begin{align}
\Norm{(\cT^n-\cT_\infty^n)(Z)}{1}^2~\leq~\sum_{i=1}^r|\trace{Y_i^*Z}|^2\,\lambda_i^{2n}~.\label{qDBboundWithLambdas}
\end{align}
\end{proposition}
\begin{proof}
One can apply the preceding general steps to the map $\cB_\sigma(X):=\sqrt{\sigma}X\sqrt{\sigma}$ and the inner product $\braket{Y}{X}:=\trace{Y^*X}$. Then it remains to show that, for $A:=(\cT^n-\cT_\infty^n)(Z)$,
$$\Norm{A}{1}^2~\leq~\braket{A}{\cB_\sigma^{-1}(A)}~=~\trace{A^*\sigma^{-1/2}A\sigma^{-1/2}}~.$$

 To see this inequality holds in fact for all $A\in\cM_d$, use the polar decomposition and let $U\in\cM_d$ be a unitary such that $UA$ is positive-semidefinite. Then cyclicity of the trace and two applications of the Cauchy-Schwarz inequality give:
\begin{align*}
\Norm{A}{1}^2&=~\left|{\rm tr}\left[UA\right]\right|^2~=~\left|{\rm tr}\left[(\sigma^{1/4}U\sigma^{1/4})(\sigma^{-1/4}A\sigma^{-1/4})\right]\right|^2\\
&\leq~{\rm tr}\left[\sigma^{1/4}U\sigma^{1/2}U^*\sigma^{1/4}\right]\,{\rm tr}\left[\sigma^{-1/4}A^*\sigma^{-1/2}A\sigma^{-1/4}\right]\\
&=~{\rm tr}\left[U\sigma^{1/2}U^*\sigma^{1/2}\right]\,{\rm tr}\left[A^*\sigma^{-1/2}A\sigma^{-1/2}\right]\\
&\leq~\sqrt{{\rm tr}\left[U\sigma U^*\right]{\rm tr}\left[\sigma^{1/2}UU^*\sigma^{1/2}\right]}\,{\rm tr}\left[A^*\sigma^{-1/2}A\sigma^{-1/2}\right]\\
&=~\trace{\sigma}\,{\rm tr}\left[A^*\sigma^{-1/2}A\sigma^{-1/2}\right]\\
&=~{\rm tr}\left[A^*\sigma^{-1/2}A\sigma^{-1/2}\right]~.
\end{align*}
\end{proof}

Detailed balance of a quantum map $\cT$ w.r.t.~certain other maps $(\cB=\Omega_\sigma^k)^{-1}$ has been defined in \cite{temme1}  so that the family $(\Omega_\sigma^k)^{-1}$ includes the map $\cB_\sigma$ from above. These detailed balance conditions also result in bounds that look essentially like Equation~\eqref{qDBboundWithLambdas}, except that in this more general case the $Y_i$ should be orthonormal in the sense that $\trace{Y_i^*\cB(Y_j)}=\delta_{ij}$. For a proof, note that Equation~\eqref{inconvl2bound} holds generally, and the proof of Lemma 5 in \cite{temme1} shows $\Norm{(\cT^n-\cT_\infty^n)(Z)}{1}^2\leq\Norm{(\cT^n-\cT_\infty^n)(v)}{2,\cB^{-1}}^2$ (the right-hand-side of the last inequality is a $\chi^2$-divergence as considered in \cite{temme1} only if $\cT$ has merely one eigenvalue of modulus $1$, however).

For classical detailed balanced Markov chains the analog of the convergence bound Inequality~\eqref{qDBboundWithLambdas}, which looks very similar in this setting \cite{mixing}, is often used for demonstrating the upper bound in cutoff results (cf.~\cite{cutoff,group} for an over overview and references). In this setting, most commonly the evolution $\cT$ leads to a unique fixed point $\sigma$ (often the maximally mixed state), so that the asymptotic evolution would simply be the ``projection onto the fixed point'', i.e.~$\cT_\infty^n(X)=\sigma\,\trace{X}$ for $n\geq1$. Of course, for Proposition \ref{quantdetbalwitheigenvaluesprop} to be useful one also needs knowledge about the normalized eigenvectors $Y_i$.
\section{Conclusions}
The conceptual innovation of this article is to provide a framework within which eigenvalue estimates can be derived for Markov chains. We apply this framework to study the relation between the spectrum of a transition map and the speed of convergence of the resulting Markov chain. Our approach yields a significant improvement on the wide-spread and important convergence estimates based on Jordan and Schur normal forms. In the analysis of the sensitivity of the stationary states of a Markov chain our bounds can be used to improve existing stability results and our methods yield strong resolvent estimates for Markov transition maps. On the purely mathematical side the main contribution of this article is to bound the quantity $\Norm{z^n}{W/mW}$ for a given polynomial $m$, which is essentially a version of Kreiss' matrix theorem with given spectral data.

\begin{acknowledgements}
We acknowledge financial support by the Elite Network of Bavaria (ENB) project QCCC, the CHIST-ERA/BMBF project CQC, the Marie-Curie project QUINTYL and the Alfried Krupp von Bohlen und Halbach-Stiftung
\end{acknowledgements}

\bibliographystyle{abbrv}

\appendix

\section{An upper bound on a single Blaschke factor}\label{Blaschke}
For completeness we prove the following short lemma.
\begin{lemma}
Let $\abs{\lambda}<c\leq1$ then
\begin{align*}
\sup_{\abs{z}=c}\left|\frac{1-\bar{\lambda}z}{z-\lambda}\right|=\frac{1-\abs{\lambda}c}{c-\abs{\lambda}}.
\end{align*}
\end{lemma}
\begin{proof}
We rewrite the absolute value on the left hand side using the fact that $\abs{a}^2=a\bar{a}$ for all $a\in\mathbb{C}$. This gives
\begin{align*}
\left|\frac{1-\bar{\lambda}z}{z-\lambda}\right|^2=\frac{(1-\abs{\lambda}c)^2+2\abs{\lambda}c-2\Re(\lambda\bar{z})}{(c-\abs{\lambda})^2+2\abs{\lambda}c-2\Re(\lambda\bar{z})}.
\end{align*}
Note now that for $0<\beta<\alpha$ and $0\leq x$ we have
\begin{align*}
\frac{\alpha+x}{\beta+x}\leq\frac{\alpha}{\beta}.
\end{align*}
Hence,
\begin{align*}
\frac{(1-\abs{\lambda}c)^2+2\abs{\lambda}c-2\Re(\lambda\bar{z})}{(c-\abs{\lambda})^2+2\abs{\lambda}c-2\Re(\lambda\bar{z})}\leq\frac{(1-\abs{\lambda}c)^2}{(c-\abs{\lambda})^2}.
\end{align*}
Finally, we note that the supremum is attained for $z=\frac{c}{\abs{\lambda}}\lambda$.
\end{proof}
\end{document}